\documentclass[12pt,letterpaper]{article}
\usepackage{epsfig,rotating,setspace,latexsym,amsmath,epsf,amssymb,bm,theorem}
\usepackage{cite,appendix}

\usepackage{amsfonts}

\title{Ergodic Secret Alignment\thanks{This work was supported by NSF Grants CCF 04-47613, CCF 05-14846, CNS 07-16311 and CCF 07-29127 and CCF 09-64645,
and presented in part at the 47th Annual Allerton Conference on
Communications, Control, and Computing, Monticello, IL, September
2009 \cite{allerton09}, and the IEEE International Conference on
Communications, Cape Town, South Africa, May 2010 \cite{ESA}.}}

\author{Raef Bassily \qquad Sennur Ulukus \\
\normalsize Department of Electrical and Computer Engineering\\
\normalsize University of Maryland, College Park, MD 20742 \\
\normalsize {\it bassily@umd.edu} \qquad {\it ulukus@umd.edu}}

\newtheorem{Theo}{Theorem}

\newenvironment{proof}[1]{\medskip\par\noindent{\bf Proof:\,}\,#1}{{\mbox{\,$\blacksquare$}\par}}

\begin{document}
\date{}
\maketitle
\begin{abstract}
In this paper, we introduce two new achievable schemes for the
fading multiple access wiretap channel (MAC-WT). In the model that
we consider, we assume that perfect knowledge of the state of all
channels is available at all the nodes in a causal fashion. Our
schemes use this knowledge together with the time varying nature
of the channel model to align the interference from different
users at the eavesdropper perfectly in a one-dimensional space
while creating a higher dimensionality space for the interfering
signals at the legitimate receiver hence allowing for better
chance of recovery. While we achieve this alignment through signal
scaling at the transmitters in our first scheme (scaling based
alignment (SBA)), we let nature provide this alignment through the
ergodicity of the channel coefficients in the second scheme
(ergodic secret alignment (ESA)). For each scheme, we obtain the
resulting achievable secrecy rate region. We show that the secrecy
rates achieved by both schemes scale with SNR as
$\frac{1}{2}\log(\mbox{SNR})$. Hence, we show the sub-optimality
of the i.i.d. Gaussian signaling based schemes with and without
cooperative jamming by showing that the secrecy rates achieved
using i.i.d. Gaussian signaling with cooperative jamming do not
scale with SNR. In addition, we introduce an improved version of
our ESA scheme where we incorporate cooperative jamming to achieve
higher secrecy rates. Moreover, we derive the necessary optimality
conditions for the power control policy that maximizes the secrecy
sum rate achievable by our ESA scheme when used solely and with
cooperative jamming.
\end{abstract}
\newpage
\section{Introduction}
The notion of information theoretic secrecy was first introduced
by Shannon in his seminal work \cite{Sh}. Applying the notion of
information theoretic secrecy to channel models with single
transmitter, single receiver, and single eavesdropper (wiretapper)
was pioneered by Wyner \cite{wyn}, Csiszar and Korner \cite{ck},
and Leung-Yan-Cheong and Hellman \cite{Hellman}. Wyner \cite{wyn},
introduced the wiretap channel where it is assumed that the
received signal by the eavesdropper is a degraded version of the
signal received by the legitimate receiver. For his model, Wyner
established the secrecy capacity region, which is defined as the
region of all simultaneously achievable rates and
equivocation-rates. In \cite{ck}, the secrecy capacity region was
established for the general case where the eavesdropper's channel
is not necessarily a degraded version of the main receiver's
channel. In particular, it was shown that to achieve the secrecy
capacity region of the single user wiretap channel, channel
prefixing may be necessary. In channel prefixing, an auxiliary
random variable serves as the input of an artificially created
prefix channel, whose output is used as the input to the original
wiretap channel. In \cite{Hellman}, the authors showed that,
through plain Gaussian signaling alone, i.e., without channel
prefixing, one can achieve the secrecy capacity of the Gaussian
wiretap channel.

The multiple access wiretap channel (MAC-WT) was introduced in
\cite{MAC21}. In MAC-WT, multiple users wish to have secure
communication with a single receiver, in the presence of a passive
eavesdropper. References \cite{MAC21} and \cite{MAC22} focus on
the Gaussian MAC-WT, and provide achievable schemes based on
Gaussian signaling. Reference \cite{MAC22} goes further than plain
Gaussian signaling and introduces a technique (on top of Gaussian
signaling) that uses the power of a non-transmitting node in
jamming the eavesdropper. This technique is called
\emph{cooperative jamming}. Cooperative jamming is indeed a
channel prefixing technique where specific choices are made for
the auxiliary random variables \cite{ekrem_ulukus_chapter}. In
addition, cooperative jamming is the first significant application
of channel prefixing in a multi-user Gaussian wiretap channel that
improves over plain Gaussian signaling. More recently, reference
\cite{MAC13} showed that for a certain class of Gaussian MAC-WT,
one can achieve through Gaussian signaling a secrecy rate region
that is within 0.5 bits of the secrecy capacity region.
Consequently, there has been some expectation that secrecy
capacity may be obtained for Gaussian MAC-WT through i.i.d.
Gaussian signaling, potentially with Gaussian channel prefixing.

However, a notable shortcoming of these Gaussian signaling based
achievable schemes is that rates obtained using them do not scale
with the signal-to-noise ratio (SNR). In other words, the total
number of degrees of freedom (DoF) for the MAC-WT achieved using
these schemes is zero. This observation led to the belief that
these schemes, and hence Gaussian signaling (with or without
channel prefixing), may be sub-optimal. This belief is made
certain as a direct consequence of the results on the secure DoF
of Gaussian interference networks that were obtained in several
papers, e.g., in \cite{int_algn3}, \cite{int_algn4},
\cite{struc1}, \cite{struc3}, and \cite{real_int_al}. In
particular, in each of \cite{int_algn3} and \cite{int_algn4}, it
was shown that positive secure DoF is achievable for a class of
vector Gaussian interference channels (i.e., time-varying channels
where channel state information is known non-causally) which in
turn implies that positive secure DoF is achievable for the vector
Gaussian MAC-WT. In \cite{struc1} and \cite{struc3}, it was shown
that through structured coding (e.g., lattice coding), it is
possible to achieve positive DoF for a class of scalar (i.e.,
non-time-varying) Gaussian channels with interference that
contains the Gaussian MAC-WT. More recently, in
\cite{real_int_al}, an achievable secrecy rate region for the
$K$-user Gaussian MAC-WT was obtained by incorporating a new
alignment technique known as real interference alignment. This
technique performs on a single real line and exploits the
properties of real numbers to align interference in time-invariant
channels.

Fading Gaussian MAC-WT was first considered in \cite{MAC23} where
the Gaussian signaling and cooperative jamming schemes which were
originally proposed in \cite{MAC21} and \cite{MAC22} are extended
to the fading MAC-WT. Using these schemes, \cite{MAC23} gave
achievable ergodic sum secrecy rates for the fading MAC-WT.
Similar to the non-fading setting, these achievable ergodic
secrecy rates do not scale with the average SNRs. In this paper,
we propose two new achievable schemes for the fading Gaussian
MAC-WT. Our first achievable scheme, the \emph{scaling based
alignment} (SBA) scheme, is based on code repetition with proper
scaling of transmitted signals. In particular, transmitters repeat
their symbols in two {\it consecutive} symbol instants.
Transmitters further scale their transmit signals with the goal of
creating a full-rank channel matrix at the main receiver and a
unit-rank channel matrix at the eavesdropper, in every two
consecutive time instants. These coordinated actions create a
two-dimensional space for the signal received by the legitimate
receiver, while sustaining the interference in a
single-dimensional space at the eavesdropper. In other words, code
repetition with proper scaling of the transmit signals at each
transmitter {\it aligns} the received signals at the eavesdropper
perfectly making it difficult for the eavesdropper to decode both
messages. Consequently, we obtain a new achievable secrecy rate
region for the two-user fading MAC-WT.

In another recent work \cite{Nazer}, it was shown that in a fading
interference channel, by code repetition over {\it properly
chosen} time instants, one can perfectly cancel interference at
each receiver so that the resulting individual rates scale as
$\frac{1}{2}\log(\mbox{SNR})$. Thus, the rate reduction by a
factor of $\frac{1}{2}$ comes with the benefit of perfect
interference cancellation. In this paper, we extend the ergodic
interference alignment concept to a secrecy context and we propose
another achievable scheme which we call {\it ergodic secret
alignment} (ESA). In the SBA scheme, code repetition is done over
two consecutive time instants, while in the ESA scheme, we
carefully choose the time instants over which we do code
repetition such that the received signals are aligned favorably at
the legitimate receiver while they are aligned unfavorably at the
eavesdropper. In particular, given some time instant with the
vector of the main receiver channel coefficients and the vector of
the eavesdropper channel coefficients given by $\textbf{h}=[h_{1}
~h_{2}]^T$ and $\textbf{g}=[g_{1} ~g_{2}]^T$, respectively, if
$X_1$ and $X_2$ are the symbols transmitted in this time instant
by users $1$ and $2$, respectively, our objective, roughly
speaking, is to determine the channel gains we should wait for to
transmit $X_1$ and $X_2$ again. In this paper, we show that, in
order to maximize achievable secrecy rates, we should wait for a
time instant in which the main receiver channel coefficients are
$[h_{1} ~-h_{2}]^T$ and the eavesdropper channel coefficients are
$[g_{1} ~g_{2}]^T$. Consequently, we obtain another achievable
secrecy rate region for the two-user fading MAC-WT.

For both proposed schemes, we show that the resulting secrecy
rates scale with SNR. Specifically, the achievable secrecy sum
rate scales as $1/2\log(\mbox{SNR})$. Moreover, we show that the
secrecy rates achieved through i.i.d. Gaussian signaling with
cooperative jamming in fading MAC-WT do not scale with SNR. The
significance of these results is that, they show that indeed
neither plain i.i.d. Gaussian signaling nor i.i.d. Gaussian
signaling with cooperative jamming is optimal for the fading
MAC-WT, and that, for high SNRs, one can achieve higher secrecy
rates by aligning interference perfectly in the eavesdropper MAC
while reducing, or cancelling, interference at the main receiver
MAC using some coordinated actions at both transmitters that
involve code repetition, i.e., a form of time-correlated (non
i.i.d.) signaling.

In fact, the achievable rate region using the second scheme, the
ESA scheme, involves two significant improvements over the one
achieved by the SBA scheme when the channel coefficients are
circularly symmetric complex Gaussian random variables.  First,
the expressions for achievable rates by the SBA scheme involve
products of the squared magnitudes of the channel coefficients.
The squared magnitudes of the channel coefficients are exponential
random variables and hence multiplying them will intuitively make
the small values of their product occur with higher probability
and the large values occur with lower probability. This in effect
reduces the achievable rates by the SBA scheme. On the other hand,
the achievable secrecy rates by the ESA scheme do not have this
drawback. In other words, by code repetition, the SBA scheme
creates two (not perfectly) correlated MAC channels to the main
receiver and two perfectly correlated MAC channels to the
eavesdropper, while the ESA scheme creates an orthogonal MAC
channel to the main receiver and two perfectly correlated MAC
channels to the eavesdropper. This fact leads to higher achievable
secrecy rates by the ESA scheme. The second improvement of the ESA
scheme with respect to the SBA scheme is that the average power
constraints associated with the ESA scheme do not involve any
channel coefficients whereas those associated with the SBA scheme
involve the gains of the eavesdropper channel which in turn result
in inefficient use of transmit powers.

In addition, we introduce an improved version of our second scheme
in which we use cooperative jamming on top of the ESA scheme to
achieve higher secrecy rates. Moreover, since the rate expressions
achieved by the ESA scheme (with and without cooperative jamming)
and their associated average power constraints are simpler than
their counterparts in the SBA scheme, we derive the necessary
conditions on the optimal power allocations that maximize the sum
secrecy rate achieved by the ESA scheme when used alone and when
used together with cooperative jamming. Since the achievable
secrecy sum rate, in general, is not a concave function in the
power allocation policy, the solution of such optimization problem
may not be unique. Hence, we obtain a power allocation policy that
satisfies the necessary (but not necessarily sufficient) KKT
conditions of optimality.

Finally, we provide numerical examples that illustrate the scaling
of the sum rates achieved by the proposed schemes with SNR and the
saturation of the secrecy sum rate achieved by the i.i.d. Gaussian
signaling scheme with cooperative jamming. We also give numerical
examples for the secrecy sum rates achieved by the ESA scheme with
and without cooperative jamming when power control is used.

\section{System Model}
We consider the two-user fading multiple access channel with an
external eavesdropper. Transmitter $k$ chooses a message $W_k$
from a set of equally likely messages
$\mathcal{W}_k=\{1,...,2^{2nR_k}\}$, $k=1,2$. Every transmitter
encodes its message into a codeword of length $2n$ symbols. The
channel output at the intended receiver and the eavesdropper are
given by
\begin{align}
  Y =& h_{1}X_{1}+h_{2}X_{2}+N\label{A1} \\
  Z =& g_{1}X_{1}+g_{2}X_{2}+N'\label{A2}
\end{align}
where, for $k=1,2$, $X_{k}$ is the input signal at transmitter
$k$, $h_{k}$ is the channel coefficient between transmitter $k$
and the intended receiver, $g_{k}$ is the channel coefficient
between transmitter $k$ and the eavesdropper. We assume a fast
fading scenario where the channel coefficients randomly vary from
one symbol to another in i.i.d. fashion. Also, we assume the
independence of all channel coefficients $h_{1}$, $h_{2}$,
$g_{1}$, and $g_{2}$. Each of the channel coefficients is a
circularly symmetric complex Gaussian random variable with
zero-mean. The variances of $h_{k}$ and $g_{k}$ are
$\sigma_{h_{k}}^2$ and $\sigma_{g_{k}}^2$, respectively. Hence,
$|h_{k}|^2$ and $|g_{k}|^2$ are exponentially distributed random
variables with mean $\sigma_{h_{k}}^2$ and $\sigma_{g_{k}}^2$,
respectively. Moreover, we assume that all the channel
coefficients are known to all the nodes in a causal fashion. In
(\ref{A1})-(\ref{A2}), $N$ and $N'$ are the independent Gaussian
noises at the intended receiver and the eavesdropper,
respectively, and are i.i.d. (in time) circularly symmetric
complex Gaussian random variables with zero-mean and
unit-variance. Moreover, we have the usual average power
constraints
\begin{align}\label{const1}
   E[|X_k|^2]&\leq \bar{P}_k,\quad k=1,2.
\end{align}
\section{Previously Known Results}\label{known_results}
Here we summarize previously known results that are relevant to
our development. For the general discrete-time memoryless MAC-WT,
the best known achievable secrecy rate region \cite{MAC21},
\cite{MAC22}, \cite{ekrem_ulukus_chapter} is given by the convex
hull of all rate pairs $(R_1,R_2)$ satisfying
\begin{align}
  R_1 &\leq I(V_1;Y|V_2)-I(V_1;Z) \label{ach1}\\
  R_2 &\leq I(V_2;Y|V_1)-I(V_2;Z) \label{ach2}\\
  R_1+R_2 &\leq I(V_1,V_2;Y)-I(V_1,V_2;Z) \label{achsum}
\end{align}
where the distribution $p(x_1,x_2,v_1,v_2,y,z)$ factors as
$p(v_1)p(x_1|v_1)p(v_2)p(x_2|v_2)p(y,z|x_1,x_2)$.

Known secrecy rate regions for the Gaussian MAC-WT can be obtained
from these expressions by appropriate selections for the involved
random variables. For instance, the Gaussian signaling based
achievable rates proposed in \cite{MAC21} are obtained by choosing
$X_1=V_1$ and $X_2=V_2$, i.e., no channel prefixing, and by
choosing $X_1$ and $X_2$ to be Gaussian with full power. On the
other hand, cooperative jamming based achievable rates proposed in
\cite{MAC22} are obtained by choosing $X_1=V_1+T_1$ and
$X_2=V_2+T_2$, and then by choosing $V_1, V_2, T_1, T_2$ to be
independent Gaussian random variables \cite{ekrem_ulukus_chapter}.
Here, $V_1$ and $V_2$ carry messages, while $T_1$ and $T_2$ are
jamming signals. The powers of $(V_1, T_1)$ and $(V_2, T_2)$
should be chosen to satisfy the power constraints of users 1 and
2, respectively. These selections yield the following achievable
rate region for the Gaussian MAC-WT \cite{MAC22}
\begin{align}
R_1&\leq \log\left(1+\frac{|h_1|^2P_1}{1+|h_1|^2Q_1+|h_2|^2Q_2}\right)-\log\left(1+\frac{|g_1|^2P_1}{1+|g_1|^2Q_1+|g_2|^2(P_2+Q_2)}\right)\\
R_2&\leq \log\left(1+\frac{|h_2|^2P_2}{1+|h_1|^2Q_1+|h_2|^2Q_2}\right)-\log\left(1+\frac{|g_2|^2P_2}{1+|g_1|^2(P_1+Q_1)+|g_2|^2Q_2}\right)\\
R_1+R_2&\leq
\log\left(1+\frac{|h_1|^2P_1+|h_2|^2P_2}{1+|h_1|^2Q_1+|h_2|^2Q_2}\right)-\log\left(1+\frac{|g_1|^2P_1+|g_2|^2P_2}{1+|g_1|^2Q_1+|g_2|^2Q_2}\right)
\end{align}
where the powers of the signals must satisfy
\begin{align}
&P_k+Q_k \leq \bar{P}_k, \quad k=1,2
\end{align}
where $P_k$ and $Q_k$ are the transmission and jamming powers,
respectively, of user $k$.

The ergodic secrecy rate region achieved by Gaussian signaling and
cooperative jamming for the fading MAC-WT can be expressed
similarly by simply including expectations over fading channel
states \cite{MAC23}
\begin{align}
&R_1\leq
E_{\textbf{h},\textbf{g}}\left\{\log\left(1+\frac{|h_1|^2P_1}{1+|h_1|^2Q_1+|h_2|^2Q_2}\right)-\log\left(1+\frac{|g_1|^2P_1}{1+|g_1|^2Q_1+|g_2|^2(P_2+Q_2)}\right)\right\}
\label{r1CJ}\\
&R_2\leq
E_{\textbf{h},\textbf{g}}\left\{\log\left(1+\frac{|h_2|^2P_2}{1+|h_1|^2Q_1+|h_2|^2Q_2}\right)-\log\left(1+\frac{|g_2|^2P_2}{1+|g_1|^2(P_1+Q_1)+|g_2|^2Q_2}\right)\right\}
\label{r2CJ}\\
&R_1+R_2\leq
E_{\textbf{h},\textbf{g}}\left\{\log\left(1+\frac{|h_1|^2P_1+|h_2|^2P_2}{1+|h_1|^2Q_1+|h_2|^2Q_2}\right)-\log\left(1+\frac{|g_1|^2P_1+|g_2|^2P_2}{1+|g_1|^2Q_1+|g_2|^2Q_2}\right)\right\}
\label{rsumCJ}
\end{align}
where $\textbf{h}=[h_1 ~h_2]^T$, $\textbf{g}=[g_1 ~g_2]^T$, and
the instantaneous powers $P_k$ and $Q_k$, which are both functions
of $\textbf{h}$ and $\textbf{g}$, satisfy
\begin{align}
E\left[P_k+Q_k\right]&\leq \bar{P}_k,~k=1,2 \label{conCJ1}
\end{align}

\section{Scaling Based Alignment (SBA)}\label{sba_scheme}
In this section, we introduce a new achievable scheme for the
fading MAC-WT. Our achievable scheme is based on code repetition
with proper scaling of the signals transmitted by each
transmitter. This is done as follows. For the channel described in
(\ref{A1})-(\ref{A2}), we use a repetition code such that each
transmitter repeats its channel input symbol twice over two {\it
consecutive} time instants. Due to code repetition, we may regard
each of the MACs to the main receiver and to the eavesdropper as a
vector MAC composed of two parallel scalar MACs, one for the
\emph{odd} time instants and the other for the \emph{even} time
instants. Consequently, we may describe the main receiver MAC
channel by the following pair of equations
\begin{align}
  Y_o &= h_{1o}X_1+h_{2o}X_2+N_o \label{Y_o1}\\
  Y_e &= h_{1e}X_1+h_{2e}X_2+N_e\label{Y_e1}
\end{align}
where, for $k=1,2$, $h_{ko},h_{ke}$ are the coefficients of the
$k$th main receiver channel in odd and even time instants,
$Y_o,Y_e$ and $N_o,N_e$ are the received signal and the noise at
the main receiver in odd and even time instants. In the same way,
we may describe the eavesdropper MAC channel by the following pair
of equations
\begin{align}
Z_o &= g_{1o}X_1+g_{2o}X_2+N'_o \label{Z_o1} \\
Z_e &= g_{1e}X_1+g_{2e}X_2+N'_e\label{Z_e1}
\end{align}
where, for $k=1,2$, $g_{ko},g_{ke}$ are the coefficients of the
$k$th eavesdropper channel in odd and even time instants,
$Z_o,Z_e$ and $N_o,N_e$ are the received signal and the noise at
the eavesdropper in odd and even time instants.

Since all the channel gains are known to all nodes in a causal
fashion, the two transmitters use this knowledge as follows. In
every symbol instant, each transmitter scales its transmit signal
with the gain of the other transmitter's channel to the
eavesdropper. That is, in every symbol duration, the first user
multiplies its channel input with $g_{2}$, the channel gain of the
second user to the eavesdropper, and the second user multiplies
its channel input with $g_{1}$, the channel gain of the first user
to the eavesdropper. Hence the main receiver MAC can be described
as
\begin{align}
  Y_o &= h_{1o}g_{2o}X_1+h_{2o}g_{1o}X_2+N_o \label{Y_o}\\
  Y_e &= h_{1e}g_{2e}X_1+h_{2e}g_{1e}X_2+N_e\label{Y_e}
\end{align}
and the eavesdropper MAC can be described as
\begin{align}
Z_o &= g_{1o}g_{2o}X_1+g_{1o}g_{2o}X_2+N'_o \label{Z_o} \\
Z_e &= g_{1e}g_{2e}X_1+g_{2e}g_{2e}X_2+N'_e\label{Z_e}
\end{align}
It is clear from (\ref{Y_o})-(\ref{Y_e}) that the space of the
received signal (without noise, i.e., high SNR) of the main
receiver over the two consecutive time instants is two-dimensional
almost surely. In other words, the channel matrix of the main
receiver vector MAC is full-rank almost surely. This is due to the
fact that the channel coefficients are drawn from continuous
bounded distributions. On the other hand, it is clear from
(\ref{Z_o})-(\ref{Z_e}) that the channel matrix of the
eavesdropper vector MAC is unit-rank. That is, the two ingredients
of our scheme, i.e., code repetition and signal scaling, let the
interfering signals at the main receiver live in a two-dimensional
space, while they {\it align} the interfering signals at the
eavesdropper in a one-dimensional space. As we will show in the
Section~\ref{DoF_section}, these properties play a central role in
achieving secrecy rates that scale with SNR. Finally, we note
that, due to signal scaling at the transmitters, the average power
constraints become
\begin{align}
  &E\left[\left(|g_{2o}|^2+|g_{2e}|^2\right)P_1\right] \leq \bar{P}_{1}\label{cons1}  \\
  &E\left[\left(|g_{1o}|^2+|g_{1e}|^2\right)P_2\right] \leq \bar{P}_{2}\label{cons2}
\end{align}
where $P_1$ and $P_2$, which are functions of the channel gains,
are the instantaneous powers of users 1 and 2, respectively.

Now, we evaluate the secrecy rate region achievable by our {\it
scaling based alignment} (SBA) scheme. Given the vector channels
(\ref{Y_o})-(\ref{Y_e}) and (\ref{Z_o})-(\ref{Z_e}), the following
secrecy rates are achievable \cite{MAC21}, \cite{MAC22},
\cite{ekrem_ulukus_chapter},
\begin{align}
  R_1 &\leq \frac{1}{2}\left[I(X_1;Y_o,Y_e|X_2,\textbf{h},\textbf{g})-I(X_1;Z_o,Z_e|\textbf{h},\textbf{g})\right] \label{r-mi-1}\\
  R_2 &\leq \frac{1}{2}\left[I(X_2;Y_o,Y_e|X_1,\textbf{h},\textbf{g})-I(X_2;Z_o,Z_e|\textbf{h},\textbf{g})\right] \label{r-mi-2}\\
  R_1+R_2 &\leq \frac{1}{2}\left[I(X_1,X_2;Y_o,Y_e|\textbf{h},\textbf{g})-I(X_1,X_2;Z_o,Z_e|\textbf{h},\textbf{g})\right] \label{r-mi-sum}
\end{align}
These expressions for achievable rates follow from
(\ref{ach1})-(\ref{achsum}) by treating channel states as outputs
at the receivers, and noting the independence of channel inputs
and channel states. We note that the factor of $1/2$ on the right
hand sides of (\ref{r-mi-1})-(\ref{r-mi-sum}) is due to repetition
coding. Now, by computing (\ref{r-mi-1})-(\ref{r-mi-sum}) with
Gaussian signals, we obtain the secrecy rate region given in the
following theorem.

\begin{Theo}\label{res_ach_reg}
For the two-user fading MAC-WT, the rate region given by all rate
pairs $(R_1,R_2)$ satisfying the following constraints is
achievable with perfect secrecy
\begin{align}
&R_1\leq\frac{1}{2}E_{\textbf{h},\textbf{g}}\left\{\log\left(1+(|h_{1o}g_{2o}|^2+|h_{1e}g_{2e}|^2)P_1\right)-\log\left(1+\frac{(|g_{1o}g_{2o}|^2+|g_{1e}g_{2e}|^2)P_1}{1+(|g_{1o}g_{2o}|^2+|g_{1e}g_{2e}|^2)P_2}\right)\right\}\label{R1}\\
&R_2\leq\frac{1}{2}E_{\textbf{h},\textbf{g}}\left\{\log\left(1+(|h_{2o}g_{1o}|^2+|h_{2e}g_{1e}|^2)P_2\right)-\log\left(1+\frac{(|g_{1o}g_{2o}|^2+|g_{1e}g_{2e}|^2)P_2}{1+(|g_{1o}g_{2o}|^2+|g_{1e}g_{2e}|^2)P_1}\right)\right\}\label{R2}\\
&R_1+R_2\leq\frac{1}{2}E_{\textbf{h},\textbf{g}}\Bigg\{\log\bigg(1+\left(|h_{1o}g_{2o}|^2+|h_{1e}g_{2e}|^2\right)P_1+\left(|h_{2o}g_{1o}|^2+|h_{2e}g_{1e}|^2\right)P_2\nonumber\\
&\hspace{5cm}+|h_{1e}h_{2o}g_{1o}g_{2e}-h_{1o}h_{2e}g_{1e}g_{2o}|^2P_1P_2\bigg)\nonumber\\
&\hspace{3.5cm}-\log\bigg(1+\left(|g_{1o}g_{2o}|^2+|g_{1e}g_{2e}|^2\right)\left(P_1+P_2\right)\bigg)\Bigg\}\label{R1R2}
\end{align}
where $\textbf{h}=[h_{1o}~h_{1e}~h_{2o}~h_{2e}]^T$,
$\textbf{g}=[g_{1o}~ g_{1e}~g_{2o}~g_{2e}]^T$, and $P_1,P_2$,
which are functions of $\textbf{h}_o=[h_{1o}~h_{2o}]^T$ and
$\textbf{g}_o=[g_{1o}~g_{2o}]^T$, are the power allocation
policies of users $1$ and $2$, respectively, that satisfy
\begin{align}
  &E\left[\left(|g_{2o}|^2+|g_{2e}|^2\right)P_1\right] \leq \bar{P}_{1}\label{ca}\\
  &E\left[\left(|g_{1o}|^2+|g_{1e}|^2\right)P_2\right] \leq \bar{P}_{2}\label{cb}
\end{align}
where $\bar{P}_1$ and $\bar{P}_2$ are the average power
constraints.
\end{Theo}

\section{Ergodic Secret Alignment (ESA)}\label{ESA_section}
After we have devised the scaling based alignment scheme, the
ergodic interference alignment scheme of Nazer {\it{et. al.}}
\cite{Nazer} inspired us to propose an improved achievable scheme.
In this section, we discuss this scheme which we call {\it ergodic
secret alignment} (ESA). The new ingredient in this scheme is to
perform repetition coding at two {\it carefully chosen} time
instances as opposed to two {\it consecutive} time instances as we
have done in Section~\ref{sba_scheme}.

For the MAC-WT described by (\ref{A1})-(\ref{A2}), we use a
repetition code in a way similar to the one in \cite{Nazer}.
Indeed, we repeat each code symbol in the time instant that holds
certain channel conditions relative to the those conditions in the
time instant where this code symbol is first transmitted. Namely,
given a time instant with the main receiver channel state vector
$\textbf{h}=[h_1 ~h_2]^T$ and the eavesdropper channel state
vector $\textbf{g}=[g_1 ~g_2]^T$, where the symbols $X_1$ and
$X_2$ are first transmitted by the two transmitters, we will solve
for the channel states $\tilde{\textbf{h}}=[\tilde{h}_1
~\tilde{h}_2]^T$ and $\tilde{\textbf{g}}=[\tilde{g}_1
~\tilde{g}_2]^T$, where these symbols should be repeated again,
such that the resulting secrecy rates achieved by Gaussian
signaling are maximized.

Due to code repetition, we may regard each of the MACs to the main
receiver and to the eavesdropper as a vector MAC composed of two
parallel scalar MACs, one for each one of the two time instants
over which the same code symbols $X_1$ and $X_2$ are transmitted.
Consequently, we may describe the main receiver MAC channel by the
following pair of equations
\begin{align}
  Y_1 &= h_1X_1+h_2X_2+N_1 \label{mainrx1}\\
  Y_2 &= \tilde{h}_1X_1+\tilde{h}_2X_2+N_2 \label{mainrx2}
\end{align}
where $Y_1,Y_2$ and $N_1,N_2$ are the received symbols and the
noise at the main receiver in the two time instants of code
repetition. In the same way, we may describe the eavesdropper MAC
channel by the following pair of equations
\begin{align}
  Z_1 &= g_1X_1+g_2X_2+N_1' \label{eve1}\\
  Z_2 &= \tilde{g}_1X_1+\tilde{g}_2X_2+N_2' \label{eve2}
\end{align}
where $Z_1,Z_2$ and $N_1',N_2'$ are the received symbols and the
noise at the eavesdropper in the two time instants of code
repetition.

In the next theorem, we give another achievable secrecy rate
region for the two-user fading MAC-WT. The achievable region is
obtained using (\ref{r-mi-1})-(\ref{r-mi-sum}) and replacing
$(Y_o,Y_e)$ and $(Z_o,Z_e)$ with $(Y_1,Y_2)$ and $(Z_1,Z_2)$,
respectively, and evaluating these expressions with Gaussian
signals, and by choosing optimal
$\tilde{\textbf{h}}=(\tilde{h}_1,\tilde{h}_2)$ and
$\tilde{\textbf{g}}=(\tilde{g}_1,\tilde{g}_2)$ to maximize the
achievable rates. In particular, we choose the repetition
instants, i.e., $\tilde{\textbf{h}}$ and $\tilde{\textbf{g}}$, in
such a way that the parallel MAC to the main receiver is the most
favorable from the main transmitter-receiver pair's point of view,
and the parallel MAC to the eavesdropper is the least favorable
from the eavesdropper's point of view. As we will show shortly as
a result of Theorem~\ref{main-res}, this optimal selection will
yield an {\it orthogonal} MAC to the main receiver and a {\it
scalar} MAC to the eavesdropper. In writing the achievable rate
expressions, we will again account for code repetition by
multiplying achievable rates by a factor of $1/2$.

\begin{Theo} \label{main-res}
For the two-user fading MAC-WT, the rate region given by all rate
pairs $(R_1,R_2)$ satisfying the following constraints is
achievable with perfect secrecy
\begin{align}
  R_1 \leq& \frac{1}{2}E_{\textbf{h},\textbf{g}}\left\{\log\left(1+2|h_1|^2P_1\right)-\log\left(1+\frac{2|g_1|^2P_1}{1+2|g_2|^2P_2}\right)\right\}\label{R1new} \\
  R_2 \leq& \frac{1}{2}E_{\textbf{h},\textbf{g}}\left\{\log\left(1+2|h_2|^2P_2\right)-\log\left(1+\frac{2|g_2|^2P_2}{1+2|g_1|^2P_1}\right)\right\}\label{R2new} \\
  R_1+R_2 \leq& \frac{1}{2}E_{\textbf{h},\textbf{g}}\left\{\log\left(1+2|h_1|^2P_1\right)+\log\left(1+2|h_2|^2P_2\right)-\log\left(1+2(|g_1|^2P_1+|g_2|^2P_2)\right)\right\}\label{Rsumnew}
\end{align}
where $P_1$ and $P_2$ are the power allocation policies of users
$1$ and $2$, respectively, and are both functions of $\textbf{h}$
and $\textbf{g}$ in general. In addition, they satisfy the average
power constraints
\begin{align}
  E[P_1] &\leq \bar{P}_1\label{con1} \\
  E[P_2] &\leq \bar{P}_2\label{con2}
\end{align}
\end{Theo}

\begin{proof}
First, consider the two vector MACs given by
(\ref{mainrx1})-(\ref{eve2}). Observe that as in \cite{Nazer},
$\tilde{\textbf{h}}$ must be chosen such that it has the same
distribution as $\textbf{h}$ and $\tilde{\textbf{g}}$ must be
chosen such that it has the same distribution as $\textbf{g}$.
Since $\textbf{h}\sim \mathcal{CN}(\textbf{0},\textbf{B}_h)$ and
$\textbf{g}\sim \mathcal{CN}(\textbf{0},\textbf{B}_g)$ where
$\textbf{B}_h={\text{diag}}(\sigma_{h_{1}}^2,\sigma_{h_{2}}^2)$
and $\textbf{B}_g=
{\text{diag}}(\sigma_{g_{1}}^2,\sigma_{g_{2}}^2)$, then
$\tilde{\textbf{h}}$ and $\tilde{\textbf{g}}$ must be in the form
$\tilde{\textbf{h}} = \textbf{U}\textbf{h}$ and
$\tilde{\textbf{g}} = \textbf{V}\textbf{g}$ where the unitary
matrices $\textbf{U}$ and $\textbf{V}$ must further be of the
form: $\textbf{U}={\text{diag}}(\exp(j\theta_1),\exp(j\theta_2))$
and $\textbf{V}={\text{diag}}(\exp(j\omega_1),\exp(j\omega_2))$
for some $\theta_1,\theta_2,\omega_1,\omega_2\in [0,2\pi)$. Then,
it follows that (\ref{mainrx1})-(\ref{eve2}) can be written as
\begin{eqnarray}
Y_1 &=& h_1X_1+h_2X_2+N_1 \label{rx1}\\
Y_2 &=& h_1e^{j\theta_1}X_1+h_2e^{j\theta_2}X_2+N_2 \label{rx2}\\
Z_1 &=& g_1X_1+g_2X_2+N_1' \label{ev1}\\
Z_2 &=& g_1e^{j\omega_1}X_1+g_2e^{j\omega_2}X_2+N_2' \label{ev2}
\end{eqnarray}
Using (\ref{r-mi-1})-(\ref{r-mi-sum}) and replacing $(Y_o,Y_e)$
and $(Z_o,Z_e)$ with $(Y_1,Y_2)$ and $(Z_1,Z_2)$, respectively,
and computing these achievable rates with Gaussian signals, we get
\begin{align}
  R_1 \leq& \frac{1}{2}E_{\textbf{h},\textbf{g}}\left\{\log\left(1+2|h_1|^2P_1\right)-\log\left(1+\frac{2|g_1|^2P_1+2(1-\cos(\omega))|g_1|^2|g_2|^2P_1P_2}{1+2|g_2|^2P_2}\right)\right\}\label{R'1} \\
  R_2 \leq& \frac{1}{2}E_{\textbf{h},\textbf{g}}\left\{\log\left(1+2|h_2|^2P_2\right)-\log\left(1+\frac{2|g_2|^2P_2+2(1-\cos(\omega))|g_1|^2|g_2|^2P_1P_2}{1+2|g_1|^2P_1}\right)\right\}\label{R'2} \\
  R_1+R_2 \leq& \frac{1}{2}E_{\textbf{h},\textbf{g}}\Big\{\log(1+2|h_1|^2P_1+2|h_2|^2P_2+2(1-\cos(\theta))|h_1|^2|h_2|^2P_1P_2)\nonumber\\
  &\hspace{1.4cm}-\log(1+2|g_1|^2P_1+2|g_2|^2P_2+2(1-\cos(\omega))|g_1|^2|g_2|^2P_1P_2)\Big\}\label{R'sum}
\end{align}
where $\theta=\theta_2-\theta_1$ and $\omega=\omega_2-\omega_1$.

Hence, the largest achievable secrecy rate region
(\ref{R'1})-(\ref{R'sum}) is attained by choosing $\theta=\pi$ and
$\omega=0$. This can be achieved by choosing $\theta_1=0$ and
$\theta_2=\pi$ and by choosing $\omega_1=\omega_2=0$.
Consequently, we have $\tilde{\textbf{h}}=[h_1 ~-h_2]^T$ and
$\tilde{\textbf{g}}=[g_1 ~g_2]^T$. By substituting these values of
$\theta$ and $\omega$ in (\ref{R'1})-(\ref{R'sum}), we obtain the
region given by (\ref{R1new})-(\ref{Rsumnew}).
\end{proof}

Therefore, when using the ergodic secret alignment technique, the
best choice for $\tilde{h}_1$ and $\tilde{h}_2$ is such that
$\tilde{\textbf{h}}$ is orthogonal to $\textbf{h}$ and that
$\|\tilde{\textbf{h}}\|=\|\textbf{h}\|$, and the best choice for
$\tilde{g}_1$ and $\tilde{g}_2$ is such that $\tilde{\textbf{g}}$
and $\textbf{g}$ are linearly dependent and that
$\|\tilde{\textbf{g}}\|=\|\textbf{g}\|$, i.e.,
$\tilde{\textbf{g}}=\textbf{g}$. This choice makes the vector MAC
between the two transmitters and the main receiver equivalent to
an orthogonal MAC, i.e., two independent single-user fading
channels, one from each transmitter to the main receiver. This
equivalent main receiver MAC channel can be expressed as
\begin{eqnarray}
  \bar{Y}_1 &=& 2h_1X_1+\bar{N}_1\label{mainrxeq1} \\
  \bar{Y}_2 &=& 2h_2X_2+\bar{N}_2\label{mainrxeq2}
\end{eqnarray}
where $\bar{Y}_1=Y_1+Y_2$, $\bar{Y}_2=Y_1-Y_2$,
$\bar{N}_1=N_1+N_2$, and $\bar{N}_2=N_1-N_2$. Note that
$\bar{N}_1$ and $\bar{N}_2$ are independent. On the other hand,
this choice makes the vector MAC between the two transmitters and
the eavesdropper equivalent to a single scalar MAC. This
equivalent eavesdropper MAC channel can be expressed as
\begin{eqnarray}
  \bar{Z}_1 &=& 2g_1X_1+2g_2X_2+\bar{N}_1'\label{eveeq1} \\
  \bar{Z}_2 &=& \bar{N}_2'\label{eveeq2}
\end{eqnarray}
where $\bar{Z}_1=Z_1+Z_2$, $\bar{Z}_2=Z_1-Z_2$,
$\bar{N}_1'=N_1'+N_2'$, and $\bar{N}_2'=N_1'-N_2'$. Note again
that $\bar{N}_1$ and $\bar{N}_2$ are independent. Note that, here,
the second component of the eavesdropper's vector MAC is useless
for her (i.e., leaks no further information than the first
component) as it contains only noise. This selection of the
repetition channel state yields a most favorable setting for the
main receiver, and a least favorable setting for the eavesdropper.

\section{Degrees of Freedom}\label{DoF_section}
In this section, we show that the secrecy sum rates achieved by
our schemes scale with SNR as $1/2\log(\mbox{SNR})$ and that the
secrecy sum rate achieved by the cooperative jamming scheme given
in \cite{MAC23} does not scale with SNR. What we give here are
rigorous proofs for intuitive results. Since by looking at
(\ref{R1R2}) and (\ref{Rsumnew}), one can note that, if we assume
that $\bar{P}_1=\bar{P}_2=P$, then if we take $P_1=P_2=P$, as $P$
becomes large, roughly speaking, 
in (\ref{R1R2}) the first term inside the expectation grows as
$\log(P^2)$ while the second term grows as $\log(P)$ and hence the
overall expression grows as $1/2\log(P)$; and similarly, in
(\ref{Rsumnew}), all three terms inside the expectation grow as
$\log(P)$ and hence the overall expression grows as $1/2\log(P)$.
In the same way, by considering the secrecy sum rate achieved by
the cooperative jamming scheme given in (\ref{rsumCJ}), then by
referring to the power allocation policies given in \cite{MAC23},
one can also roughly say that for all channel states, as the
available average power goes to infinity, the overall expression
converges to a constant.

For simplicity, we assume symmetric average power constraints for
all schemes, i.e., we set $\bar{P}_1=\bar{P}_2=P$ in
(\ref{ca})-(\ref{cb}), (\ref{con1})-(\ref{con2}), and
(\ref{conCJ1}). We also assume that all channel gains are drawn
from continuous bounded distributions and that all channel gains
have finite variances. Let $R_{s}$ be the achievable secrecy sum
rate, then the total number of achievable secure DoF, $\eta$, is
defined as
\begin{align}
    \eta&\triangleq \lim_{P\rightarrow \infty}\frac{R_s}{\log(P)}\label{eta}
\end{align}
We start by the DoF analysis of our proposed schemes, i.e., the
SBA scheme and the ESA scheme, where we show that the sum secrecy
rates obtained by these schemes achieve $1/2$ secure DoF, then we
provide a rigorous proof for the fact that the scheme of
\cite{MAC23} which is based on i.i.d. Gaussian signaling with
cooperative jamming achieves a secrecy sum rate that does not
scale with SNR, i.e., achieves zero secure DoF.

\subsection{Secure DoF with the SBA Scheme}\label{DoF_SBA}
We make the following choices for the power allocation policies
$P_1$ and $P_2$ of the SBA scheme. We set
$P_1=\frac{1}{2\sigma_{g_{2}}^2}P$,
$P_2=\frac{1}{2\sigma_{g_{1}}^2}P$. It can be verified that these
choices satisfy the power constraints (\ref{ca})-(\ref{cb}).
Denoting the expression inside the expectation in (\ref{R1R2}) by
$f_P(\textbf{h},\textbf{g})$, the secrecy sum rate achieved using
the SBA scheme can be written as
\begin{align}
R_{s}&=\frac{1}{2}E_{\textbf{h},\textbf{g}}\left\{f_P(\textbf{h},\textbf{g})\right\}\label{SumRate}
\end{align}
Hence, the total achievable secure DoF is given by
\begin{align}
&\eta=\frac{1}{2}\lim_{P\rightarrow
\infty}E_{\textbf{h},\textbf{g}}\left[\frac{f_P(\textbf{h},\textbf{g})}{\log(P)}\right]\label{eta1}
\end{align}

Now, we show that, for the two-user fading MAC-WT, a total number
of secure DoF $\eta=1/2$ is achievable with the SBA scheme.
Towards this end, it suffices to show that the order of the limit
and the expectation in (\ref{eta1}) can be reversed. To do this,
we make use of Lebesgue dominated convergence theorem. Now, we
note that for large enough $P$,
$\frac{f_P(\textbf{h},\textbf{g})}{\log(P)}$ is upper bounded by
$\psi(\textbf{h},\textbf{g})$ where
\begin{align}
\psi(\textbf{h},\textbf{g})=&4+2\left(\log\left(1+\frac{1}{\sigma_{g_{1}}^2}\right)+\log\left(1+\frac{1}{\sigma_{g_{2}}^2}\right)\right)+\log\left(1+\frac{\sigma_{g_{1}}^2+\sigma_{g_{2}}^2}{\sigma_{g_{1}}^2\sigma_{g_{2}}^2}\right)\nonumber\\
&+3\left(\sum_{k=1}^2\log(1+|h_{ko}|^2)+\sum_{k=1}^2\log(1+|h_{ke}|^2)\right)\nonumber\\
&+4\left(\sum_{k=1}^2\log(1+|g_{ko}|^2)+\sum_{k=1}^2\log(1+|g_{ke}|^2)\right)
\end{align}
Hence, using the fact that all channel gains have finite variances
together with Jensen's inequality, we have
\begin{align}
E_{\textbf{h},\textbf{g}}\left[\psi(\textbf{h},\textbf{g})\right]&<\infty
\end{align}
Thus, by the dominated convergence theorem, we have
\begin{align}
\lim_{P\rightarrow
\infty}E_{\textbf{h},\textbf{g}}\left[\frac{f_P(\textbf{h},\textbf{g})}{\log(P)}\right]=E_{\textbf{h},\textbf{g}}\left[\lim_{P\rightarrow
\infty}\frac{f_P(\textbf{h},\textbf{g})}{\log(P)}\right]=1
\end{align}
Hence, from (\ref{eta1}), we have $\eta = 1/2$.

\subsection{Secure DoF with the ESA Scheme}
We show that the ESA scheme achieves $\eta=1/2$ secure DoF in the
two-user fading MAC-WT. Here, we also use a constant power
allocation policy for the ESA scheme where we set $P_1=P_2=P$ for
all channel states. Clearly, this constant policy satisfies the
average power constraints (\ref{con1})-(\ref{con2}). Denoting the
expression inside the expectation in (\ref{Rsumnew}) by
$\tilde{f}_P(\textbf{h},\textbf{g})$, the achievable secrecy sum
rate, $R_{s}$ is given by
\begin{align}
R_{s}&=\frac{1}{2}E_{\textbf{h},\textbf{g}}\left\{\tilde{f}_P(\textbf{h},\textbf{g})\right\}\label{Rs_ESA}
\end{align}
Hence, the total achievable secure DoF is given by
\begin{align}
&\eta=\frac{1}{2}\lim_{P\rightarrow
\infty}E_{\textbf{h},\textbf{g}}\left[\frac{\tilde{f}_P(\textbf{h},\textbf{g})}{\log(P)}\right]\label{eta2}
\end{align}
We note that for large enough $P$,
$\frac{\tilde{f}_{P}(\textbf{h},\textbf{g})}{\log(P)}\leq
\tilde{\psi}(\textbf{h},\textbf{g})$ where
\begin{align}
\tilde{\psi}(\textbf{h},\textbf{g})&=6+\log\left(1+2|h_1|^2\right)+\log\left(1+2|h_2|^2\right)+\log\left(1+2\left(|g_1|^2+|g_2|^2\right)\right)\label{psi_ESA}
\end{align}
Again, using the fact that all channel gains have finite variances
together with Jensen's inequality, we have
\begin{align}
E_{\textbf{h},\textbf{g}}\left[\tilde{\psi}(\textbf{h},\textbf{g})\right]&<\infty
\end{align}
Then, by the dominated convergence theorem, we have
\begin{align}
\lim_{P\rightarrow
\infty}E_{\textbf{h},\textbf{g}}\left[\frac{\tilde{f}_P(\textbf{h},\textbf{g})}{\log(P)}\right]=E_{\textbf{h},\textbf{g}}\left[\lim_{P\rightarrow
\infty}\frac{\tilde{f}_P(\textbf{h},\textbf{g})}{\log(P)}\right]=1
\end{align}
Hence, from (\ref{eta2}), we have $\eta = 1/2$.

\subsection{Secure DoF with i.i.d. Gaussian Signaling with CJ}
We consider the secrecy sum rate achieved by Gaussian signaling
with cooperative jamming (CJ) \cite{MAC23} in the fading MAC-WT
and show that this achievable rate does not scale with SNR. We
start with the secrecy sum rate given by the right hand side of
(\ref{rsumCJ}). According to the optimal power allocation policy
described in \cite{MAC23}, for $k=1,2$, we cannot have $P_k>0$ and
$Q_k>0$ simultaneously. Moreover, no transmission occurs when
$|h_1|\leq |g_1|$ and $|h_2|\leq |g_2|$. Consequently, according
to the relative values of the channel gains
$(|h_1|,|h_2|,|g_1|,|g_2|)$, there are three different cases left
for the instantaneous secrecy sum rate achieved using the optimum
power allocation where we omitted the case where $|h_1|\leq |g_1|$
and $|h_2|\leq |g_2|$ since no transmission is allowed.

\textbf{Case 1:} $(\textbf{h},\textbf{g})\in\mathcal{D}_1$ where
$\mathcal{D}_1=\big\{(\textbf{h},\textbf{g}):|h_1|>|g_1|,|h_2|>|g_2|\big\}$.
Consequently, $Q_1=Q_2=0$. Thus, the instantaneous secrecy sum
rate, $R_{s}(\textbf{h},\textbf{g})$, can be written as
\begin{align}
R_{s}(\textbf{h},\textbf{g})=
   \log\left(\frac{1+|h_1|^2P_1+|h_2|^2P_2}{1+|g_1|^2P_1+|g_2|^2P_2}\right)\label{case1}
\end{align}
We can upper bound $R_{s}(\textbf{h},\textbf{g})$ as
\begin{align}
  R_{s}(\textbf{h},\textbf{g})&\leq \log\left(1+\frac{|h_1|^2}{|g_1|^2}+\frac{|h_2|^2}{|g_2|^2}\right)\leq\log\left(1+\frac{|h_1|^2}{|g_1|^2}\right)+\log\left(1+\frac{|h_2|^2}{|g_2|^2}\right)\label{casei}
\end{align}

\textbf{Case 2:} $(\textbf{h},\textbf{g})\in\mathcal{D}_2$ where
$\mathcal{D}_2=\big\{(\textbf{h},\textbf{g}):|h_1|>|g_1|,|h_2|<|g_2|\big\}$.
Consequently, $Q_1=P_2=0$. Thus, the instantaneous secrecy sum
rate, $R_{s}(\textbf{h},\textbf{g})$, can be written as
\begin{align}
   &R_{s}(\textbf{h},\textbf{g})=\log\left(\frac{1+|h_1|^2P_1+|h_2|^2Q_2}{1+|g_1|^2P_1+|g_2|^2Q_2}\right)+\log\left(\frac{1+|g_2|^2Q_2}{1+|h_2|^2Q_2}\right)\label{case2}
\end{align}
We can upper bound $R_{s}(\textbf{h},\textbf{g})$ as
\begin{align}
    R_{s}(\textbf{h},\textbf{g})\leq 1+\log\left(1+\frac{|h_1|^2}{|g_1|^2}\right)
    +\log\left(1+\frac{|g_2|^2}{|h_2|^2}\right)\label{caseii}
\end{align}

\textbf{Case 3:} $(\textbf{h},\textbf{g})\in\mathcal{D}_3$ where
$\mathcal{D}_3=\big\{(\textbf{h},\textbf{g}):|h_1|<|g_1|,|h_2|>|g_2|\big\}$.
Consequently, $P_1=Q_2=0$. Thus, the instantaneous secrecy sum
rate, $R_{s}(\textbf{h},\textbf{g})$, can be written as
\begin{align}
 &R_{s}(\textbf{h},\textbf{g})=\log\left(\frac{1+|h_1|^2Q_1+|h_2|^2P_2}{1+|g_1|^2Q_1+|g_2|^2P_2}\right)+\log\left(\frac{1+|g_1|^2Q_1}{1+|h_1|^2Q_1}\right)\label{case3}
\end{align}
We can upper bound $R_{s}(\textbf{h},\textbf{g})$ as
\begin{align}
  R_{s}(\textbf{h},\textbf{g})\leq 1+\log\left(1+\frac{|h_2|^2}{|g_2|^2}\right)
  +\log\left(1+\frac{|g_1|^2}{|h_1|^2}\right)\label{caseiii}
\end{align}

Now, since the instantaneous sum rate is zero outside
\mbox{$\mathcal{D}_1\cup\mathcal{D}_2\cup\mathcal{D}_3$}, then
from (\ref{casei}), (\ref{caseii}), and (\ref{caseiii}), the
ergodic secrecy sum rate, $R_s$, can be upper bounded as follows
\begin{align}
R_{s}&\leq\int_{\mathcal{D}_1}\left(\log\left(1+\frac{|h_1|^2}{|g_1|^2}\right)+\log\left(1+\frac{|h_2|^2}{|g_2|^2}\right)\right)d\textbf{F}\nonumber\\
&\hspace{0.1cm}+\int_{\mathcal{D}_2}\left(1+\log\left(1+\frac{|h_1|^2}{|g_1|^2}\right)+\log\left(1+\frac{|g_2|^2}{|h_2|^2}\right)\right)d\textbf{F}\nonumber\\
&\hspace{0.1cm}+\int_{\mathcal{D}_3}\left(1+\log\left(1+\frac{|h_2|^2}{|g_2|^2}\right)+\log\left(1+\frac{|g_1|^2}{|h_1|^2}\right)\right)d\textbf{F}\label{Rsum22}
\end{align}
where
\begin{align}
d\textbf{F}=\prod_{k=1}^2f(|h_k|^2)f(|g_k|^2)d|h_k|^2d|g_k|^2
\end{align}
where, for $k=1,2$, $f(|h_k|^2)$ and $f(|g_k|^2)$ are the density
functions of $|h_k|^2$ and $|g_k|^2$, respectively. Now, since
$E[|h_k|^2]<\infty,\hspace{0.05cm}E[|g_k|^2]<\infty$ for $k=1,2$,
$|\int_0^1\log(x)dx|=\log(e)<\infty$,
$|\int_0^1\log(1+x)dx|=2-\log(e)<\infty$, and
$f(|h_k|^2),f(|g_k|^2)$ are continuous and bounded for $k=1,2$, it
follows that each of the three integrals in the above expression
is finite. Hence, we have $R_{s}<\infty$, and that $R_{s}$ is
bounded from above by a constant. Thus, from definition
(\ref{eta}) of the achievable secure DoF, $\eta$, we have
\begin{align}
&\eta=\lim_{P\rightarrow\infty}\frac{R_{s}}{\log(P)}=0
\end{align}

\section{ESA Scheme with Cooperative Jamming}
The result given in Theorem~\ref{main-res} can be strengthened by
adding the technique of cooperative jamming to the ESA scheme of
Section~\ref{ESA_section}. We refer to the resulting scheme as
ESA/CJ. This is done through Gaussian channel prefixing as
discussed in Section~\ref{known_results}. In particular, we choose
the channel inputs in (\ref{mainrx1})-(\ref{eve2}) to be
$X_1=V_1+T_1$ and $X_2=V_2+T_2$, and then choose $V_1, V_2, T_1,
T_2$ to be independent Gaussian random variables. Here, $V_1$ and
$V_2$ carry messages, while $T_1$ and $T_2$ are jamming signals.
The powers of $(V_1, T_1)$ and $(V_2, T_2)$ should be chosen to
satisfy the average power constraints of users 1 and 2,
respectively. These selections when made in the ESA scheme yield
the following achievable rate region which, through appropriate
power control strategy (see Section~\ref{optESACJ}), can be made
strictly larger than the region given in Theorem~\ref{main-res},
\begin{align}
 R_1 &\leq \frac{1}{2}E_{\textbf{h},\textbf{g}}\Bigg\{\log\left(1+\frac{2|h_1|^2P_1}{1+2|h_1|^2Q_1}\right)-\log\left(1+\frac{2|g_1|^2P_1}{1+2|g_1|^2Q_1+2|g_2|^2(P_2+Q_2)}\right)\Bigg\}\label{RCJ1}\\
 R_2 &\leq \frac{1}{2}E_{\textbf{h},\textbf{g}}\Bigg\{\log\left(1+\frac{2|h_2|^2P_2}{1+2|h_2|^2Q_2}\right)-\log\left(1+\frac{2|g_2|^2P_2}{1+2|g_1|^2(P_1+Q_1)+2|g_2|^2Q_2}\right)\Bigg\}\label{RCJ2}\\
 R_1+R_2 &\leq\frac{1}{2}E_{\textbf{h},\textbf{g}}\Bigg\{\log\left(1+\frac{2|h_1|^2P_1}{1+2|h_1|^2Q_1}\right)+\log\left(1+\frac{2|h_2|^2P_2}{1+2|h_2|^2Q_2}\right)\nonumber\\
  &\hspace{1.85cm}-\log\left(1+\frac{2(|g_1|^2P_1+|g_2|^2P_2)}{1+2(|g_1|^2Q_1+|g_2|^2Q_2)}\right)\Bigg\}\label{RsumCJ}
\end{align}
where, for $k=1,2,$ $P_k$ and $Q_k$ are the transmission and
jamming powers, respectively, of user $k$, and are both functions
of $\textbf{h}$ and $\textbf{g}$ in general. In addition, they
satisfy the average power constraints
\begin{align}
  E[P_k+Q_k] &\leq \bar{P}_k,~k=1,2\label{conCJ2}
\end{align}

\section{Maximizing Secrecy Sum Rate of the ESA Scheme}\label{optESA}
In this section, we consider the problem of maximizing the secrecy
sum rate achieved by the ESA scheme as a function of the power
allocations $P_1$ and $P_2$ of users $1$ and $2$, respectively.
For notational convenience, we replace $2|h_k|^2$ and $2|g_k|^2$
in the achievable rates (\ref{R1new})-(\ref{Rsumnew}) by $h_k$ and
$g_k$, respectively. Then, we define
$\textbf{h}\triangleq[h_1\quad h_2]^T$ and
$\textbf{g}\triangleq[g_1\quad g_2]^T$. The achievable secrecy sum
rate is given by
\begin{align}
R_{s}
=&\frac{1}{2}E_{\textbf{h},\textbf{g}}\{\log\left(1+h_1P_1\right)+\log\left(1+h_2P_2\right)-\log\left(1+g_1P_1+g_2P_2\right)\}
\end{align}
We can write the optimization problem as
\begin{align}
  \max\hspace{0.5cm}&E_{\textbf{h},\textbf{g}}\{\log\left(1+h_1P_1\right)+\log\left(1+h_2P_2\right)-\log\left(1+g_1P_1+g_2P_2\right)\}\label{objective}\\
  {\text{s.t.}}\hspace{0.5cm}&E_{\textbf{h},\textbf{g}}\left[P_k(\textbf{h},\textbf{g})\right] \leq \bar{P}_k,\quad k=1,2\label{opt_const1} \\
  \hspace{1cm} &P_k(\textbf{h},\textbf{g}) \geq 0,\quad k=1,2,\quad \forall \textbf{h},\textbf{g}\label{opt_const2}
\end{align}
The necessary KKT optimality conditions are
\begin{align}
    \frac{h_1}{1+h_1P_1}-\frac{g_1}{1+g_1P_1+g_2P_2}-(\lambda_1-\mu_1)&=0\label{lagr1}\\
    \frac{h_2}{1+h_2P_2}-\frac{g_2}{1+g_1P_1+g_2P_2}-(\lambda_2-\mu_2)&=0\label{lagr2}
\end{align}
It should be noted here that (\ref{lagr1})-(\ref{lagr2}) are only
necessary conditions for the optimal power allocations $P_1$ and
$P_2$ since the objective function, i.e., the achievable secrecy
sum rate, is not concave in $(P_1,P_2)$ in general.

For each channel state, we distinguish between three non-zero
forms that the solution $(P_1,P_2)$ of (\ref{lagr1})-(\ref{lagr2})
may take. First, if $P_1>0$ and $P_2>0$, then $\mu_1=\mu_2=0$.
Hence $(P_1,P_2)$ is the positive common root of the following two
quadratic equations
\begin{align}
    h_1\left(1+g_2P_2\right)-g_1&=\lambda_1\left(1+h_1P_1\right)\left(1+g_1P_1+g_2P_2\right)\label{optpospower1}\\
    h_2\left(1+g_1P_1\right)-g_2&=\lambda_2\left(1+h_2P_2\right)\left(1+g_1P_1+g_2P_2\right)\label{optpospower2}
\end{align}
Since it is hard to find a simple closed-form solution for the
above system of equations, we solve this system numerically and
obtain the positive common root $(P_1,P_2)$. Secondly, if $P_1>0$
and $P_2=0$, then $\mu_1=0$. Hence, from (\ref{lagr1}), $P_1$ is
given by
\begin{align}
P_1&=\frac{1}{2}\left(\sqrt{\left(\frac{1}{g_1}-\frac{1}{h_1}\right)^2+\frac{4}{\lambda_1}\left(\frac{1}{g_1}-\frac{1}{h_1}\right)}-\left(\frac{1}{g_1}+\frac{1}{h_1}\right)\right)\label{power1onlypos}
\end{align}
Thirdly, if $P_1=0$ and $P_2>0$, then $\mu_2=0$. Hence, from
(\ref{lagr2}), $P_2$ is given by
\begin{align}
P_2&=\frac{1}{2}\left(\sqrt{\left(\frac{1}{g_2}-\frac{1}{h_2}\right)^2+\frac{4}{\lambda_2}\left(\frac{1}{g_2}-\frac{1}{h_2}\right)}-\left(\frac{1}{g_2}+\frac{1}{h_2}\right)\right)\label{power2onlypos}
\end{align}
From conditions (\ref{lagr1})-(\ref{lagr2}), we can derive the
following necessary and sufficient conditions for the positivity
of the optimal power allocation policies:
\begin{align}
P_1>0,\qquad&\text{if and only if}\quad h_1
-\frac{g_1}{\left(1+g_2P_2\right)}>\lambda_1\label{necsufcondpospower1}\\
P_2>0,\qquad&\text{if and only if}\quad
h_2-\frac{g_2}{\left(1+g_1P_1\right)}>\lambda_2\label{necsufcondpospower2}
\end{align}
Consequently, according to conditions
(\ref{necsufcondpospower1})-(\ref{necsufcondpospower2}), we can
divide the set of all possible channel states into $7$ partitions
such that in each partition the solution $(P_1,P_2)$ will either
have one of the three forms stated above or will be zero. Hence,
the power allocation policy $(P_1,P_2)$ that
satisfies (\ref{lagr1})-(\ref{lagr2}) and
(\ref{opt_const1})-(\ref{opt_const2}) can be fully described in
$7$ different cases of the channel gains. The details of such cases
are given in Appendix~\ref{app1}.

\section{Maximizing Secrecy Sum Rate of the ESA/CJ Scheme}\label{optESACJ}
In this section, we consider the problem of maximizing the
achievable secrecy sum rate as a function in the power allocation
policies $P_1$ and $P_2$ when cooperative jamming technique is
used on top of the ESA scheme. Again, for notational convenience,
we replace $2|h_k|^2$ and $2|g_k|^2$ in the achievable rates
(\ref{RCJ1})-(\ref{RsumCJ}) by $h_k$ and $g_k$, respectively.
Then, we define $\textbf{h}\triangleq[h_1\quad h_2]^T$ and
$\textbf{g}\triangleq[g_1\quad g_2]^T$. In this case, the
optimization problem is described as
\begin{align}
  \max&\hspace{0.5cm}E_{\textbf{h},\textbf{g}}\big\{\log\left(1+h_1(P_1+Q_1)\right)+\log\left(1+h_2(P_2+Q_2)\right)\nonumber\\
  &\hspace{1.6cm}-\log\left(1+g_1(P_1+Q_1)+g_2(P_2+Q_2)\right)+\log\left(1+g_1Q_1+g_2Q_2\right)\nonumber\\
  &\hspace{1.6cm}-\log\left(1+h_1Q_1\right)-\log\left(1+h_2Q_2\right)\big\}\label{objectiveCJ}\\
  {\text{s.t.}}&\hspace{0.5cm}E_{\textbf{h},\textbf{g}}\left[P_k(\textbf{h},\textbf{g})+Q_k(\textbf{h},\textbf{g})\right] \leq \bar{P}_k,\quad k=1,2\label{opt_constCJ1} \\
  &\hspace{0.5cm} P_k(\textbf{h},\textbf{g}),Q_k(\textbf{h},\textbf{g}) \geq 0,\quad k=1,2,~\forall \textbf{h},\textbf{g} \label{opt_constCJ2}
\end{align}

We first show that, at any fading state, splitting a user's power
into transmission and jamming is suboptimal, i.e., an optimum
power allocation policy must not have $P_k>0$ and $Q_k>0$
simultaneously. We note that whether we split powers or not does
not affect the first three terms of the objective function since
we can always convert jamming power of user $k$ into transmission
power of the same user and vice versa while keeping the sum
$P_k+Q_k$ fixed. Hence, we consider the last three terms of the
sum rate. For convenience, we define
\begin{align}
S&=\log\left(1+g_1Q_1+g_2Q_2\right)-\log\left(1+h_1Q_1\right)-\log\left(1+h_2Q_2\right)
\end{align}
Consider, without loss of generality, the power allocation for
user $1$. We assume that $P_1^{\ast},Q_1^{\ast}$ is the optimum
power allocation for user $1$. We observe that the sign of
\begin{align}
\frac{\partial S}{\partial
Q_1}&=\frac{g_1}{1+g_1Q_1+g_2Q_2}-\frac{h_1}{1+h_1Q_1}\label{powersplitnotopt}
\end{align}
does not depend on $Q_1$. Consider a power allocation
$P_1=P_1^{\ast}-\alpha, Q_1=Q_1^{\ast}+\alpha$. Hence, we have
$P_1+Q_1=P_1^{\ast}+Q_1^{\ast}$ and the first three terms in the
expression of the achievable sum rate do not change. On the other
hand, if (\ref{powersplitnotopt}) is positive, any positive
$\alpha$ results in an increase in the achievable sum rate and
jamming with the same sum power is better. While, if
(\ref{powersplitnotopt}) is negative, then any negative $\alpha$
results in an increase in the achievable sum rate and transmitting
with the same sum power is better. If (\ref{powersplitnotopt}) is
zero, then the sum rate does not depend on $Q_1$ and we can set it
to zero, i.e., use the sum power in transmitting. Therefore, the
optimum power allocation will have either $P_k>0$ or $Q_k>0$, but
not both.

Suppose that $P_1,P_2,Q_1,$ and $Q_2$ are the optimal power
allocations. Then, the necessary KKT conditions satisfy
\begin{align}
&\frac{h_1}{1+h_1(P_1+Q_1)}-\frac{g_1}{1+g_1(P_1+Q_1)+g_2(P_2+Q_2)}-(\lambda_1-\mu_1)=0\label{lagrCJ1}\\
&\frac{h_2}{1+h_2(P_2+Q_2)}-\frac{g_2}{1+g_1(P_1+Q_1)+g_2(P_2+Q_2)}-(\lambda_2-\mu_2)=0\label{lagrCJ2}\\
&\frac{h_1}{1+h_1(P_1+Q_1)}-\frac{g_1}{1+g_1(P_1+Q_1)+g_2(P_2+Q_2)}+\frac{g_1}{1+g_1Q_1+g_2Q_2}\nonumber\\
&-\frac{h_1}{1+h_1Q_1}-(\lambda_1-\nu_1)=0\label{lagrCJ3}\\
&\frac{h_2}{1+h_2(P_2+Q_2)}-\frac{g_2}{1+g_1(P_1+Q_1)+g_2(P_2+Q_2)}+\frac{g_2}{1+g_1Q_1+g_2Q_2}\nonumber\\
&-\frac{h_2}{1+h_2Q_2}-(\lambda_2-\nu_2)=0\label{lagrCJ4}
\end{align}
As in Section~\ref{optESA}, we note that
(\ref{lagrCJ1})-(\ref{lagrCJ4}) are only necessary conditions for
the optimal power allocations $P_1,P_2,Q_1,$ and $Q_2$ since the
objective function, i.e., the achievable secrecy sum rate, is not
concave in $(P_1,P_2,Q_1,Q_2)$ in general. Therefore, we give
power control policies $P_1,P_2,Q_1,$ and $Q_2$ that satisfy these
necessary conditions. That is, we obtain one fixed point
$(P_1,P_2,Q_1,Q_2)$ of the Lagrangian such that
$(P_1,P_2,Q_1,Q_2)$ satisfies the constraints
(\ref{opt_constCJ1})-(\ref{opt_constCJ2}). The power allocation
policy $(P_1,P_2,Q_1,Q_2)$ that satisfies
(\ref{lagrCJ1})-(\ref{lagrCJ4}) and
(\ref{opt_constCJ1})-(\ref{opt_constCJ2}) is described in detail
in Appendix~\ref{app2}.

\section{Numerical Results}
In this section, we present some simple simulation results. We
also plot the sum secrecy rate achieved using our SBA and ESA
schemes, as well as the i.i.d. Gaussian signaling with cooperative
jamming (GS/CJ) scheme in \cite{MAC23}. First, the secrecy sum
rates achieved by the SBA and the ESA schemes scale with SNR.
Hence, these rates exceed the one achieved by the GS/CJ scheme for
high SNR. Second, the secrecy sum rate achieved by the ESA scheme
is larger than the one achieved by the SBA scheme for all SNR.

In our first set of simulations, we use a rudimentary power
allocation policy for our SBA and ESA schemes. For the SBA scheme,
we first note, from (\ref{R1R2}), that the secrecy sum rate
achieved can be expressed as a nested expectation as
\begin{align}
R_s=&\frac{1}{2}E_{\textbf{h}_o,\textbf{g}_o}\Bigg\{E_{\textbf{h}_e,\textbf{g}_e}\Bigg[\log\bigg(1+\left(|h_{1o}g_{2o}|^2+|h_{1e}g_{2e}|^2\right)P_1+\left(|h_{2o}g_{1o}|^2+|h_{2e}g_{1e}|^2\right)P_2\nonumber\\
&\hspace{4cm}+|h_{1e}h_{2o}g_{1o}g_{2e}-h_{1o}h_{2e}g_{1e}g_{2o}|^2P_1P_2\bigg)\nonumber\\
&\hspace{3cm}-\log\bigg(1+\left(|g_{1o}g_{2o}|^2+|g_{1e}g_{2e}|^2\right)\left(P_1+P_2\right)\bigg)\Bigg]\Bigg\}
\label{nested-exp}
\end{align}
where $\textbf{h}_o=[h_{1o}~h_{2o}]^T$,
$\textbf{h}_e=[h_{1e}~h_{2e}]^T$,
$\textbf{g}_o=[g_{1o}~g_{2o}]^T$, and
$\textbf{g}_e=[g_{1e}~g_{2e}]^T$. For those channel gains
$\textbf{h}_o,\textbf{g}_o$ for which the inner expectation with
respect to $\textbf{h}_e,\textbf{g}_e$ is negative, we set
$P_1=P_2=0$. Otherwise, we set
$P_1=\frac{1}{2\sigma_{g}^2}\bar{P}_1$ and
$P_2=\frac{1}{2\sigma_{g}^2}\bar{P}_2$. Note that turning off the
powers for some values of the channel gains
$\textbf{h}_o,\textbf{g}_o$ is possible since $P_1$ and $P_2$ are
functions of $\textbf{h}_o$ and $\textbf{g}_o$. Secondly, note
that, if a power allocation satisfies the average power
constraints, then the modified power allocation where the powers
are turned off at some channel states, also satisfies the power
constraints. For the ESA scheme, we first note, from
(\ref{Rsumnew}), that the achievable secrecy sum rate is
\begin{align}
R_s=&\frac{1}{2}E_{\textbf{h},\textbf{g}}\left\{\log\left(1+2|h_1|^2P_1\right)+\log\left(1+2|h_2|^2P_2\right)-\log\left(1+2(|g_1|^2P_1+|g_2|^2P_2)\right)\right\}
\end{align}
In this case, we set $P_1=P_2=0$ for those values of channel gains
for which the difference inside the expectation is negative.
Otherwise, we set $P_1=\bar{P}_1$ and $P_2=\bar{P}_2$. Again,
turning the powers off does not violate power constraints for a
power allocation scheme which already satisfies the power
constraints. For the GS/CJ scheme, we use the power allocation
scheme described in \cite{MAC23}.

\begin{figure}[t]
\centerline{\includegraphics[width=4.5in]{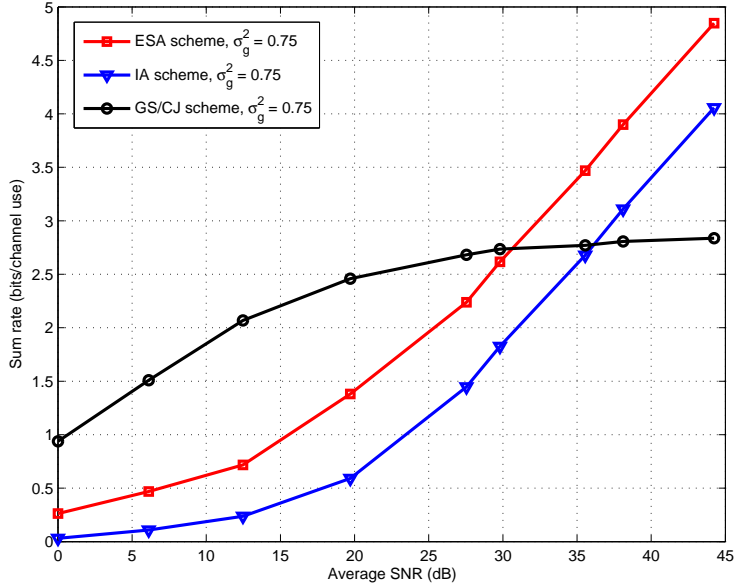}}
\caption{Achievable secrecy sum-rates of the scaling based
alignment scheme (SBA scheme) of this paper, the ergodic secret
alignment scheme (ESA scheme) of this paper, and the i.i.d.
Gaussian signaling with cooperative jamming scheme (GS/CJ scheme)
of \cite{MAC23}, as function of the SNR for two different values
of mean eavesdropper channel gain, $\sigma_g^2$.}\label{3schemes}
\end{figure}

In Figure \ref{3schemes}, the secrecy sum rate achieved by each of
the three schemes is plotted versus the average SNR that we define
as $\frac{1}{2}(\bar{P}_1+\bar{P}_2)$. In all simulations, we set
$\sigma_{h_{1}}^2=\sigma_{h_{2}}^2=1.0$, we also take
$\sigma_{g_{1}}^2=\sigma_{g_{2}}^2=0.75$.
\begin{figure}[t]
\centerline{\includegraphics[width=4.5in]{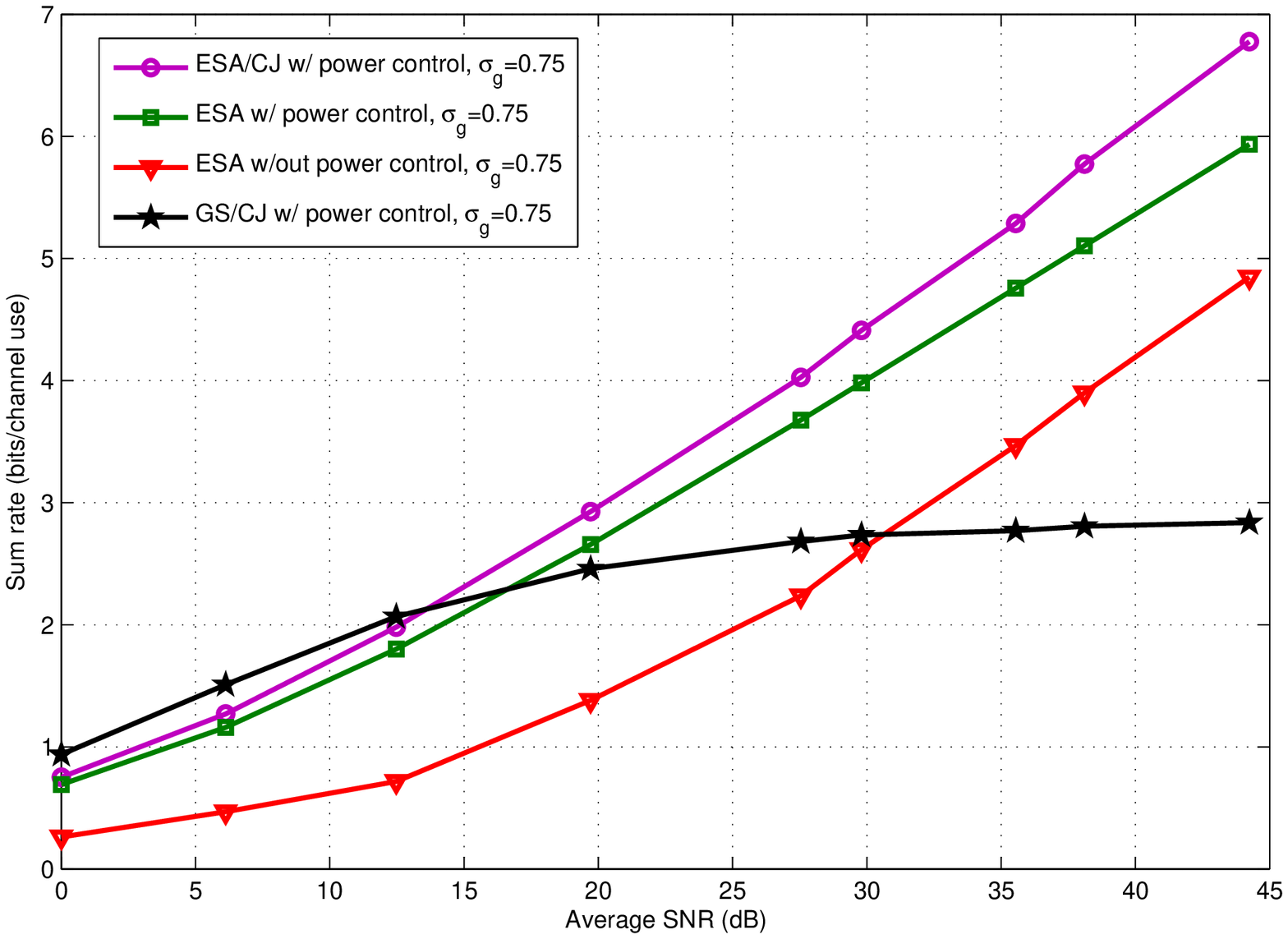}}
\caption{Achievable secrecy sum rates for the ergodic secret
alignment scheme (ESA scheme) of this paper, with and without
power control, the ergodic secret alignment with cooperative
jamming scheme (ESA/CJ scheme) of this paper with power control,
and the i.i.d. Gaussian signaling with cooperative jamming scheme
(GS/CJ scheme) of \cite{MAC23}, as function of the SNR for two
different values of mean eavesdropper channel gain,
$\sigma_g^2$.}\label{ESA_PC_CJ}
\end{figure}

Next, in Figure~\ref{ESA_PC_CJ}, we plot secrecy sum rates
achievable with constant power allocation together with secrecy
sum rates achievable with power control for the ESA scheme with
and without cooperative jamming.

\section{Conclusions}
In this paper, we proposed two new achievable schemes for the
fading multiple access wiretap channel. Our first scheme, the
scaling based alignment (SBA) scheme, lets the interfering signals
at the main receiver live in a two-dimensional space, while it
aligns the interfering signals at the eavesdropper in a
one-dimensional space. We obtained the secrecy rate region
achieved by this scheme. We showed that the secrecy rates achieved
by this scheme scale with SNR as $1/2\log(\mbox{SNR})$, i.e., a
total of $1/2$ secure DoF is achievable in the two-user fading
MAC-WT. We also showed that the secrecy sum rate achieved by the
i.i.d. Gaussian signaling with cooperative jamming scheme does not
scale with SNR, i.e., the achievable secure DoF is zero. As a
direct consequence, we showed the sub-optimality of the i.i.d.
Gaussian signaling based schemes with or without cooperative
jamming in the fading MAC-WT.

Our second scheme, the ergodic secret alignment (ESA) scheme, is
inspired by the ergodic interference alignment technique. In this
scheme each transmitter repeats its symbols over carefully chosen
time instants such that the interfering signals from the
transmitters are aligned favorably at the main receiver while they
are aligned unfavorably at the eavesdropper. We obtained the
secrecy rate region achieved by this scheme and showed that, as in
the scaling based alignment scheme, the secrecy sum rate achieved
by the ergodic secret alignment scheme scales with SNR as
$1/2\log(\mbox{SNR})$. In addition, we introduced an improved
version of our ESA scheme where cooperative jamming is used as an
additional ingredient to achieve higher secrecy rates. Moreover,
since the rate expressions achieved with the SBA scheme seem
complicated, while the rate expressions achieved with the two
versions of the ESA scheme (with and without cooperative jamming)
are more amenable for optimization of power allocations, we
derived the necessary conditions for the optimal power allocation
that maximizes the secrecy sum rate achieved by the ESA scheme
when used solely and when used with cooperative jamming.

\appendix
\appendixpage

\section{Power Control for the ESA Scheme}\label{app1}
Here, we discuss the cases of the power allocation policy of
Section~\ref{optESA}.

\begin{enumerate}
\item $h_1\leq\lambda_1,h_2-g_2\leq\lambda_2$\quad or\quad
$h_1-g_1\leq\lambda_1,h_2\leq\lambda_2$. In this case,
$P_1=P_2=0$. To prove this, suppose without loss of generality
that $h_1\leq\lambda_1,h_2-g_2\leq\lambda_2$. We note that
$h_1\leq\lambda_1$ implies that
$h_1-\frac{g_1}{\left(1+g_2P_2\right)}\leq\lambda_1$ which, using
(\ref{necsufcondpospower1}), implies that $P_1=0$. Hence, from
(\ref{necsufcondpospower2}), we must also have $P_2=0$. In the
same way, we can show that when
$h_1-g_1\leq\lambda_1,h_2\leq\lambda_2$, we also must have
$P_1=P_2=0$.

\item $h_1\leq\lambda_1,h_2-g_2>\lambda_2$. In this case, $P_1=0$
and $P_2>0$ where $P_2$ is \mbox{given by (\ref{power2onlypos}).}
As in the previous case, $h_1\leq\lambda_1$, using
(\ref{necsufcondpospower1}), implies that $P_1=0$. Hence, from
(\ref{necsufcondpospower2}), we must have $P_2>0$.

\item $h_1-g_1>\lambda_1,h_2\leq\lambda_2$. In this case, $P_1>0$
and $P_2=0$ where $P_1$ is \mbox{given by (\ref{power1onlypos}).}
This case is the same as the previous one with roles of users $1$
and $2$ interchanged.

\item
$\lambda_1<h_1\leq\lambda_1+g_1,\lambda_2<h_2\leq\lambda_2+g_2$.
In this case, the solution $(P_1,P_2)$ may not be unique. Namely,
we either have $P_1>0$ and $P_2>0$, or we have $P_1=P_2=0$. This
is due to the following facts. It is easy to see that $P_1=P_2=0$
satisfies $h_1-\frac{g_1}{\left(1+g_2P_2\right)}\leq\lambda_1$ and
$h_2-\frac{g_2}{\left(1+g_1P_1\right)}\leq\lambda_2$, i.e.,
satisfies conditions (\ref{necsufcondpospower1}) and
(\ref{necsufcondpospower2}). It is also easy to see that we can
find positive $P_1$ and $P_2$ such that
$h_1-\frac{g_1}{\left(1+g_2P_2\right)}>\lambda_1$ and
$h_2-\frac{g_2}{\left(1+g_1P_1\right)}>\lambda_2$, i.e., there
exist positive $P_1$ and $P_2$ that satisfy
(\ref{necsufcondpospower1}) and (\ref{necsufcondpospower2}). Hence
the solution $(P_1,P_2)$ may not be unique. It remains to show
that we cannot have $P_1>0, P_2=0$ or $P_1=0,P_2>0$. Suppose
without loss of generality that $P_1>0,P_2=0$. Hence, we have
$h_1-\frac{g_1}{\left(1+g_2P_2\right)}=h_1-g_1\leq\lambda_1$ which
implies that $P_1=0$ which is a contradiction. Thus, we cannot
have $P_1>0,P_2=0$. In the same way, it can be shown that we
cannot have $P_1=0,P_2>0$. Hence, we obtain our power allocation
policy for this case as follows. We examine the solution of
equations (\ref{optpospower1})-(\ref{optpospower2}), if it yields
a real and non-negative solution $(P_1,P_2)$\footnote{Note that
there is at most one such common root for these two quadratic
equations.}, then we take it as our solution $(P_1,P_2)$ for this
case. Otherwise, we set $P_1=P_2=0$.

\item $\lambda_1<h_1\leq\lambda_1+g_1,h_2-g_2>\lambda_2$. In this
case, we must have $P_2>0$. However, we either have $P_1>0$\quad
or \quad$P_1=0$. This can be shown as follows. We note that
$h_2-g_2>\lambda_2$ implies that
$h_2-\frac{g_2}{\left(1+g_1P_1\right)}>\lambda_2$ for any $P_1\geq
0$. Hence, by (\ref{necsufcondpospower2}), we must have $P_2>0$.
However, we either have $P_1>0$ or $P_1=0$ depending on whether
the value of $P_2$ satisfies
$h_1-\frac{g_1}{\left(1+g_2P_2\right)}>\lambda_1$ or not. We
obtain our power allocation policies as follows. We first solve
(\ref{optpospower1})-(\ref{optpospower2}), if this yields a real
and non-negative solution $(P_1,P_2)$, then we take it to be the
power allocation values for this case. Otherwise, we set $P_1=0$
and $P_2$ is obtained from (\ref{power2onlypos}).

\item $h_1-g_1>\lambda_1,\lambda_2<h_2\leq\lambda_2+g_2$. By the
symmetry between this case and the previous case, we must have
$P_1>0$ while we either have $P_2>0$\quad or \quad$P_2=0$. We
obtain our power allocation policies in a fashion similar to that
of case~$4$ and case~$5$. In particular, we first solve
(\ref{optpospower1})-(\ref{optpospower2}), if this yields a real
and non-negative solution $(P_1,P_2)$, then we take it to be the
power allocation values for this case. Otherwise, we set $P_2=0$
and $P_1$ is obtained from (\ref{power1onlypos}).

\item $h_1-g_1>\lambda_1,h_2-g_2>\lambda_2$. Here, we must have
$P_1>0$ and $P_2>0$. This is due to the fact that
$h_1-g_1>\lambda_1$ and $h_2-g_2>\lambda_2$ imply that
$h_1-\frac{g_1}{\left(1+g_2P_2\right)}>\lambda_1$ and
$h_2-\frac{g_2}{\left(1+g_1P_1\right)}>\lambda_2$, respectively.
\mbox{Hence, from
(\ref{necsufcondpospower1})-(\ref{necsufcondpospower2}), we must
have $P_1>0$ and $P_2>0$.} The values of $P_1$ and $P_2$ are given
by the positive common root $(P_1,P_2)$ of
(\ref{optpospower1})-(\ref{optpospower2}) which, in this case,
have only one positive common root.
\end{enumerate}

\section{Power Control for the ESA/CJ Scheme}\label{app2}

Here, we discuss the power allocation policy of
Section~\ref{optESACJ}.

For each channel state, since splitting power between transmission
and jamming is sub-optimal, we can distinguish between five
non-zero forms that the solution $(P_1,P_2,Q_1,Q_2)$ of
(\ref{lagrCJ1})-(\ref{lagrCJ4}) may take. First, if $P_1>0,P_2>0$
and $Q_1=Q_2=0$, then $\mu_1=\mu_2=0$. Hence, from
(\ref{lagrCJ1})-(\ref{lagrCJ2}), we conclude that $(P_1,P_2)$ is
the positive common root of equations
(\ref{optpospower1})-(\ref{optpospower2}) which are found in
Section~\ref{optESA} and are rewritten here:
\begin{align}
    h_1\left(1+g_2P_2\right)-g_1&=\lambda_1\left(1+h_1P_1\right)\left(1+g_1P_1+g_2P_2\right)\label{optpospower_again1}\\
    h_2\left(1+g_1P_1\right)-g_2&=\lambda_2\left(1+h_2P_2\right)\left(1+g_1P_1+g_2P_2\right)\label{optpospower_again2}
\end{align}
This root can be obtained through numerical solution. Secondly, if
$P_1>0,Q_2>0$ and $P_2=Q_1=0$, then $\mu_1=\nu_2=0$. Hence, from
(\ref{lagrCJ1}) and (\ref{lagrCJ3}), we conclude that $(P_1,Q_2)$
is the positive common root of
\begin{align}
h_1\left(1+g_2Q_2\right)-g_1&=\lambda_1\left(1+h_1P_1\right)\left(1+g_1P_1+g_2Q_2\right)\label{optposP1}\\
g_2g_1P_1&=\lambda_2\left(1+g_2Q_2\right)\left(1+g_1P_1+g_2Q_2\right)\label{optposQ2}
\end{align}
which can also be obtained through numerical solution. Thirdly, if
$P_2>0,Q_1>0$ and $P_1=Q_2=0$, then $\mu_2=\nu_1=0$. Hence, from
(\ref{lagrCJ2}) and (\ref{lagrCJ4}), we conclude that $(P_2,Q_1)$
is the positive common root of
\begin{align}
h_2\left(1+g_1Q_1\right)-g_2&=\lambda_2\left(1+h_2P_2\right)\left(1+g_1Q_1+g_2P_2\right)\label{optposP2}\\
g_1g_2P_2&=\lambda_1\left(1+g_1Q_1\right)\left(1+g_1Q_1+g_2P_2\right)\label{optposQ1}
\end{align}
which again can be obtained through numerical solution. The fourth
non-zero form of $(P_1,P_2,Q_1,Q_2)$ is when $P_1>0$ and
$P_2=Q_1=Q_2=0$, then $\mu_1=0$. Hence, from (\ref{lagrCJ1}),
$P_1$ is given by (\ref{power1onlypos}) which is found in
Section~\ref{optESA} and will be repeated here for convenience:
\begin{align}
P_1&=\frac{1}{2}\left(\sqrt{\left(\frac{1}{g_1}-\frac{1}{h_1}\right)^2+\frac{4}{\lambda_1}\left(\frac{1}{g_1}-\frac{1}{h_1}\right)}-\left(\frac{1}{g_1}+\frac{1}{h_1}\right)\right)\label{power1onlypos_again}
\end{align}
The last non-zero form of $(P_1,P_2,Q_1,Q_2)$ is when $P_2>0$ and
$P_1=Q_1=Q_2=0$, then $\mu_2=0$. Hence, from (\ref{lagrCJ2}),
$P_2$ is given by (\ref{power2onlypos}) in Section~\ref{optESA}
and is given here again.
\begin{align}
P_2&=\frac{1}{2}\left(\sqrt{\left(\frac{1}{g_2}-\frac{1}{h_2}\right)^2+\frac{4}{\lambda_2}\left(\frac{1}{g_2}-\frac{1}{h_2}\right)}-\left(\frac{1}{g_2}+\frac{1}{h_2}\right)\right)\label{power2onlypos_again}
\end{align}

We obtain the following sufficient conditions on zero jamming
powers $Q_1$ and $Q_2$. By subtracting (\ref{lagrCJ3}) from
(\ref{lagrCJ1}) and subtracting (\ref{lagrCJ4}) from
(\ref{lagrCJ2}), we get
\begin{align}
\frac{h_1}{1+h_1Q_1}-\frac{g_1}{1+g_1Q_1+g_2Q_2}+\mu_1-\nu_1&=0\label{CJ5}\\
\frac{h_2}{1+h_2Q_2}-\frac{g_2}{1+g_1Q_1+g_2Q_2}+\mu_2-\nu_2&=0\label{CJ6}
\end{align}
which, by using the fact that the two users cannot be jamming
together, give the following conditions
\begin{align}
&Q_1=0,\qquad\text{if}\quad h_1>g_1 \label{sufQ1zero}\\
&Q_2=0,\qquad\text{if}\quad h_2>g_2 \label{sufQ2zero}
\end{align}
Moreover, we obtain necessary and sufficient conditions for the
positivity of power allocations in the possible
transmission/jamming scenarios in each channel state. First, when
no user jams, i.e., $Q_1=Q_2=0$, then from
(\ref{lagrCJ1})-(\ref{lagrCJ2}), we obtain the necessary and
sufficient conditions
(\ref{necsufcondpospower1})-(\ref{necsufcondpospower1}) of
Section~\ref{optESA} which we repeat here for convenience.
\begin{align}
P_1>0,\qquad&\text{if and only if}\quad h_1
-\frac{g_1}{\left(1+g_2P_2\right)}>\lambda_1\label{necsufcondpospower_again1}\\
P_2>0,\qquad&\text{if and only if}\quad
h_2-\frac{g_2}{\left(1+g_1P_1\right)}>\lambda_2\label{necsufcondpospower_again2}
\end{align}
Secondly, when user $1$ does not jam and user $2$ does not
transmit, i.e., $Q_1=P_2=0$, then from (\ref{lagrCJ1}) and
(\ref{lagrCJ3}), we can easily derive the following necessary and
sufficient conditions for the positivity of the transmission power
$P_1$ of user $1$ and the jamming power $Q_2$ of \mbox{user $2$.}
\begin{align}
P_1>0,\qquad&\text{if and only if}\quad h_1-\frac{g_1}{\left(1+g_2Q_2\right)}>\lambda_1\label{necsufcondposP1}\\
Q_2>0,\qquad&\text{if and only if}\quad
g_2-\frac{g_2}{\left(1+g_1P_1\right)}>\lambda_2\label{necsufcondposQ2}
\end{align}
Thirdly, when user $1$ does not transmit and user $2$ does not
jam, i.e., $P_1=Q_2=0$, then from (\ref{lagrCJ2}) and
(\ref{lagrCJ4}), we can similarly derive the following necessary
and sufficient conditions for the positivity of the transmission
power $P_2$ of user $2$ and the jamming power $Q_1$ of \mbox{user
$1$}.
\begin{align}
P_2>0,\qquad&\text{if and only if}\quad h_2-\frac{g_2}{\left(1+g_1Q_1\right)}>\lambda_2\label{necsufcondposP2}\\
Q_1>0,\qquad&\text{if and only if}\quad
g_1-\frac{g_1}{\left(1+g_2P_2\right)}>\lambda_1\label{necsufcondposQ1}
\end{align}

Using conditions (\ref{sufQ1zero})-(\ref{necsufcondposQ1}) given
above, the power allocation policy $(P_1,P_2,Q_1,Q_2)$ that
satisfies (\ref{lagrCJ1})-(\ref{lagrCJ4}) and
(\ref{opt_constCJ1})-(\ref{opt_constCJ2}) can be fully described
through the following cases of the channel gains.

\begin{enumerate}
\item $h_1>g_1,h_2>g_2$. In this case, we must have $Q_1=Q_2=0$.
This follows directly from (\ref{sufQ1zero})-(\ref{sufQ2zero}).
Hence, this case reduces to one of the $7$ cases given in
Section~\ref{optESA} depending on the relative values of the
channel gains and the values of $\lambda_1$ and $\lambda_2$. We
can obtain the power allocations $P_1$ and $P_2$ in the same way
described in Section~\ref{optESA}.

\item $h_1>g_1,h_2<g_2$. In this case, we must have $P_2=Q_1=0$.
This can be shown as follows. From (\ref{sufQ1zero}), we must have
$Q_1=0$. Suppose $P_2>0$. Hence, $\mu_2=0$. Since dividing power
among transmission and jamming is suboptimal, then we must have
$Q_2=0$. Since $Q_1=0$, then (\ref{CJ6}) implies
$\bar{h}_2-\bar{g}_2\geq 0$ which is a contradiction. Therefore,
$P_2=0$. The power allocations $P_1$ and $Q_2$ are obtained from
one of the following sub-cases:
\begin{enumerate}
  \item $h_1\leq\lambda_1$ \quad or\quad $h_1-g_1\leq\lambda_1,g_2\leq\lambda_2$.
  We have $P_1=Q_2=0$. To see this, note that $h_1\leq \lambda_1$ implies that $h_1-\frac{g_1}{\left(1+g_2Q_2\right)}\leq\lambda_1$.
  Hence, using (\ref{necsufcondposP1}), we must have $P_1=0$ and thus $Q_2=0$ since we cannot have a jamming user when the other user
  is not transmitting. On the other hand, if $g_2\leq\lambda_2$, then it follows from (\ref{necsufcondposQ2}) that $Q_2=0$. Hence,
  the fact that $h_1-g_1\leq\lambda_1$ together with (\ref{necsufcondposP1}) implies that $P_1=0$.
  \item $h_1-g_1>\lambda_1,g_2\leq\lambda_2$. We have $Q_2=0$ and $P_1>0$ where $P_1$ is
  \mbox{given by (\ref{power1onlypos_again}).} This can be shown to be true as follows. Since $g_2\leq\lambda_2$, then,
  using (\ref{necsufcondposQ2}), we must have $Q_2=0$. Hence, from (\ref{necsufcondposP1}) and the fact that $h_1-g_1>\lambda_1$ in
  this case, we must have $P_1>0$.
  \item $\lambda_1<h_1\leq\lambda_1+g_1,g_2>\lambda_2$. In this case, the solution $(P_1,Q_2)$ may not be unique.
  Namely, we either have $P_1>0$ and $Q_2>0$, or we have $P_1=Q_2=0$. This is due to the following facts. It is easy to see that $P_1=Q_2=0$
  satisfies $h_1-\frac{g_1}{\left(1+g_2Q_2\right)}\leq\lambda_1$ and $g_2-\frac{g_2}{\left(1+g_1P_1\right)}\leq\lambda_2$, i.e. conditions
  (\ref{necsufcondposP1}) and (\ref{necsufcondposQ2}). It is also easy to see that we can find positive $P_1$ and $Q_2$ that satisfy
  $h_1-\frac{g_1}{\left(1+g_2Q_2\right)}>\lambda_1$ and $g_2-\frac{g_2}{\left(1+g_1P_1\right)}>\lambda_2$, i.e. conditions (\ref{necsufcondposP1})
  and (\ref{necsufcondposQ2}). Hence the solution $(P_1,Q_2)$ may not be unique. It remains to show that we cannot have $P_1>0, Q_2=0$.
  Suppose that $P_1>0$ and $Q_2=0$. Hence, we have $h_1-\frac{g_1}{\left(1+g_2Q_2\right)}=h_1-g_1\leq\lambda_1$ which, by (\ref{necsufcondposP1}),
  implies that $P_1=0$ which is a contradiction. Thus, we cannot have $P_1>0$ and $Q_2=0$. We obtain our power allocation policies for this case
  as follows. We examine the solution of equations (\ref{optposP1}) and (\ref{optposQ2}), if it yields a real and non-negative solution $(P_1,Q_2)$,
  then we take it as our solution $(P_1,Q_2)$ for this case. Otherwise, we set $P_1=Q_2=0$.
  \item $h_1-g_1>\lambda_1,g_2>\lambda_2$. Here, we must have $P_1>0$. However, we either have $Q_2>0$ or $Q_2=0$,
  i.e., the solution may not be unique. To see this, we note that $h_1-g_1>\lambda_1$ implies that $h_1-\frac{g_1}{\left(1+g_2Q_2\right)}>\lambda_2$
  for any $Q_2\geq 0$. Hence, by (\ref{necsufcondposP1}), we must have $P_1>0$. However, we either have $Q_2>0$ or $Q_2=0$ depending on whether the
  value of $P_1$ satisfies $g_2-\frac{g_2}{\left(1+g_1P_1\right)}>\lambda_1$ or not. We obtain our power allocation policy as follows.
  We first solve (\ref{optposP1}) and (\ref{optposQ2}), if this yields a real and non-negative solution $(P_1,Q_2)$, then we take it to
  be the power allocation values for this case. Otherwise, we set $Q_2=0$ and $P_1$ is obtained from (\ref{power1onlypos_again}).
\end{enumerate}

\item $h_1<g_1,h_2>g_2$. From the symmetry between this case and
the previous case, the power allocation roles can be obtained in
this case by interchanging the power allocation roles of users $1$
and $2$ in the previous case. In particular, we must have
$P_1=Q_2=0$. The power allocations $P_2$ and $Q_1$ are given by
one of the following sub-cases:
\begin{enumerate}
  \item $h_2\leq\lambda_2$\quad or\quad $g_1\leq\lambda_1,h_2-g_2\leq\lambda_2$. We have $P_2=Q_1=0$.
  \item $g_1\leq\lambda_1,h_2-g_2>\lambda_2$. We have $Q_1=0$ and $P_2>0$ where $P_2$ is \mbox{given by
  (\ref{power2onlypos_again}).}
  \item $g_1>\lambda_1,\lambda_2<h_2\leq\lambda_2+g_2$. In this case, the solution $(P_2,Q_1)$ may not be
  unique as we either have $P_2>0$ and $Q_1>0$, or have $P_1=Q_2=0$. Therefore, we obtain our power allocation policy for this case by
  numerically solving equations (\ref{optposP2}) and (\ref{optposQ1}), if we have a real and non-negative solution $(P_2,Q_1)$, then we take
  it as to be the power allocation values for this case. Otherwise, we set $P_2=Q_1=0$.
  \item $g_1>\lambda_1,h_2-g_2>\lambda_2$. Here, we must have $P_2>0$. However, we either have $Q_1>0$ or $Q_1=0$,
  i.e., the solution may not be unique. We obtain our power allocation policy as follows. We first solve (\ref{optposP2})-(\ref{optposQ1}),
  if this yields a real and non-negative solution $(P_2,Q_1)$, then we take it to be the power allocation values for this case. Otherwise,
  we set $Q_1=0$ and $P_2$ is obtained from (\ref{power2onlypos_again}).
\end{enumerate}

\item $h_1<g_1,h_2<g_2$. In this case, we have $P_2=Q_1=0$ or
$P_1=Q_2=0$. In order to see this, suppose $P_1>0$ and $P_2>0$.
Hence, $\mu_1=\mu_2=0$. Since splitting a user's power into
transmit and jamming powers is suboptimal, then we must have
$Q_1=Q_2=0$. Thus, from (\ref{CJ5}) and (\ref{CJ6}), we have
$\bar{h}_1\geq\bar{g}_1$ and $\bar{h}_2\geq\bar{g}_2$ which is a
contradiction. Therefore, we must have either $P_1=0$ or $P_2=0$.
The power allocation policy $(P_1,P_2,Q_1,Q_2)$ is given in the
following four sub-cases of channel states:
\begin{enumerate}
  \item ($h_1\leq\lambda_1$\quad or\quad $g_2\leq\lambda_2$) and ($h_2\leq\lambda_2$\quad or\quad $g_1\leq\lambda_1$).
  In this case, we have $P_1=P_2=Q_1=Q_2=0$. To see this, first, suppose that $P_2=Q_1=0$. We note that if $h_1\leq \lambda_1$ then
  $h_1-\frac{g_1}{\left(1+g_2Q_2\right)}\leq\lambda_1$. Hence, using (\ref{necsufcondposP1}), we must have $P_1=0$ and thus $Q_2=0$
  since we cannot have a jamming user when the other user is not transmitting. On the other hand, if $g_2\leq\lambda_2$, then it follows
  from (\ref{necsufcondposQ2}) that $Q_2=0$. Hence, the fact that $h_1<g_1$ together with (\ref{necsufcondposP1}) implies that $P_1=0$.
  Next, suppose that $P_1=Q_2=0$. Using the fact that $h_2\leq\lambda_2$ or $g_1\leq\lambda_1$ together with condition
  (\ref{necsufcondposP2})-(\ref{necsufcondposQ1}), we can show that $P_2=Q_1=0$. Therefore, in this case, we must have $P_1=P_2=Q_1=Q_2=0$.
  \item ($h_2\leq\lambda_2$\quad or\quad $g_1\leq\lambda_1$) and ($h_1>\lambda_1,g_2>\lambda_2$). We have $P_2=Q_1=0$.
  The solution $(P_1,Q_2)$ may not be unique. In particular, we may have $P_1>0,Q_2>0$ or have $P_1=Q_2=0$. To see this, consider the following argument.
  Using the fact that $h_2\leq\lambda_2$ or $g_1\leq\lambda_1$, then, as shown in case~$4(a)$, we conclude that we must have $P_2=Q_1=0$.
  Now, we consider the power allocation policy $(P_1,Q_2)$. We note that $P_1=Q_2=0$ satisfies conditions (\ref{necsufcondposP1}) and
  (\ref{necsufcondposQ2}). On the other hand, we can find positive $P_1$ and $Q_2$ that satisfy (\ref{necsufcondposP1}) and (\ref{necsufcondposP1}).
  Hence, the solution $(P_1,Q_2)$ may not be unique as we may have $P_1=Q_2=0$ or $P_1>0,Q_2>0$. It remains to show that we cannot have $P_1>0,Q_2=0$.
  Suppose that $P_1>0$ and $Q_2=0$. Hence, we have $h_1-\frac{g_1}{\left(1+g_2Q_2\right)}=h_1-g_1<0<\lambda_1$ which, by (\ref{necsufcondposP1}),
  implies that $P_1=0$ which is a contradiction. Thus, we cannot have $P_1>0$ and $Q_2=0$. Our power allocations $P_1$ and $Q_2$ are obtained for this
  case as follows. We solve (\ref{optposP1}) and (\ref{optposQ2}). If the solution gives a real and non-negative common root $(P_1,Q_2)$,
  we take it as our power allocation values for $P_1$ and $Q_2$. Otherwise, we set $P_1=Q_2=0$.
  \item ($h_1\leq\lambda_1$ \quad or\quad $g_2\leq\lambda_2$) and ($h_2>\lambda_2,g_1>\lambda_1$). By the symmetry between
  this case and case~$4(b)$, we have $P_1=Q_2=0$. Again in this case, the solution $(P_2,Q_1)$ may not be unique. In particular, we may have
  $P_2>0,Q_1>0$ or have $P_2=Q_1=0$. In fact, the power allocation policy in this case, can be obtained from case~$4(b)$ by interchanging
  the roles of \mbox{users $1$ and $2$.} Our power allocations $P_2$ and $Q_1$ are obtained as follows in this case. We solve
  (\ref{optposP2})-(\ref{optposQ1}). If the solution gives a real and non-negative common root $(P_2,Q_1)$, we take it as our power allocation
  values for $P_2$ and $Q_1$. Otherwise, we set $P_2=Q_1=0$.
  \item ($h_1>\lambda_1,g_2>\lambda_2$) and ($h_2>\lambda_2,g_1>\lambda_1$). Here, again the solution $(P_1,P_2,Q_1,Q_2)$
  is not unique as we may either have $P_1>0,Q_2>0,P_2=Q_1=0$, or $P_2>0,Q_1>0,P_1=Q_2=0$, or $P_1=P_2=Q_1=Q_2=0$. To see this, first, suppose that
  $P_2=Q_1=0$ and consider the power allocation policy $(P_1,Q_2)$. As in case~$4(b)$, we can show that the solution $(P_1,Q_2)$ may not be
  unique as we may have $P_1=Q_2=0$ or $P_1>0,Q_2>0$. However, as shown in case~$4(b)$, we cannot have $P_1>0,Q_2=0$. Next, suppose that
  $P_1=Q_2=0$ and consider the power allocation policy $(P_2,Q_1)$. As in case~$4(c)$, we can show that the solution $(P_2,Q_1)$ may not be
  unique as we may have $P_2=Q_1=0$ or $P_2>0,Q_1>0$. However, we cannot have $P_2>0,Q_1=0$. We obtain our allocation policy $(P_1,P_2,Q_1,Q_2)$
  as follows. Let us denote the solution of (\ref{optposP1}) and (\ref{optposQ2}) together by \emph{solution $A$} and denote the solution of
  (\ref{optposP2}) and (\ref{optposQ1}) together by \emph{solution $B$}.
      \begin{enumerate}
      \item If solution $A$ yields a real non-negative $(P_1,Q_2)$ while solution $B$ does not yield real non-negative $(P_2,Q_1)$, then we take
      $(P_1,Q_2)$ to be the power allocation values for users $1$ \mbox{and $2$}, respectively, and set $P_2=Q_1=0$.
      \item If solution $B$ yields a real non-negative $(P_2,Q_1)$ while solution $A$ does not yield real non-negative $(P_1,Q_2)$, then we take
      $(P_2,Q_1)$ to be the power allocation values for users $2$ \mbox{and $1$}, respectively, and set $P_1=Q_2=0$.
      \item If neither solution $A$ nor solution $B$ gives real non-negative common root, then we set $P_1=P_2=Q_1=Q_2=0$.
      \item If both solutions $A$ and $B$ yield a real non-negative common root, then we either choose the root given by solution $A$, i.e.,
      $(P_1,Q_2)$, and set $P_2=Q_1=0$, or choose the root given by solution $B$, i.e., $(P_2,Q_1)$, and set $P_1=Q_2=0$. We make the choice that
      maximizes the achievable \emph{instantaneous} secrecy sum rate.
      \end{enumerate}
\end{enumerate}
\end{enumerate}

\bibliographystyle{unsrt}

\end{document}